\documentclass[12pt,onecolumn,draftclsnofoot]{IEEEtran}
\usepackage{amsfonts}
\usepackage{cite}
\usepackage[cmex10]{amsmath}
\usepackage{amssymb}
\usepackage{array}
\usepackage{chemarrow}
\usepackage{times,epsfig}
\usepackage{graphicx}
\usepackage{subcaption}
\usepackage{setspace}
\usepackage{times,epsfig}
\usepackage{bbm}
\usepackage{verbatim}
\usepackage{amsmath}
\usepackage{epstopdf}
\usepackage{mathtools}

\newtheorem{remark}{\textbf{Remark}}
\newtheorem{corollary}{\textbf{Corollary}}
\newtheorem{definition}{\textbf{Definition}}

\newtheorem{lemma}{\textbf{Lemma}}

\newcolumntype{P}[1]{>{\centering\arraybackslash}p{#1}}
\newcolumntype{M}[1]{>{\centering\arraybackslash}m{#1}}

\makeatletter
\g@addto@macro{\normalsize}{%
   \setlength{\abovedisplayskip}{3pt plus 2pt minus 2pt}
   \setlength{\abovedisplayshortskip}{3pt plus 2pt minus 2pt}
   \setlength{\belowdisplayskip}{3pt plus 2pt minus 2pt}
   \setlength{\belowdisplayshortskip}{3pt plus 2pt minus 2pt}
   \setlength{\textfloatsep}{10pt plus 2pt minus 2pt}
   }
\makeatother

\begin{document}
%

\title{Modeling and Analyzing the Coexistence of Wi-Fi and LTE in Unlicensed Spectrum}

\author{ Yingzhe Li, Fran\c{c}ois Baccelli, Jeffrey G. Andrews, Thomas D. Novlan, Jianzhong Charlie Zhang \thanks{Y. Li, F. Baccelli and J. G. Andrews are with the Wireless Networking and Communications Group (WNCG), The University of Texas at Austin (email: yzli@utexas.edu, francois.baccelli@austin.utexas.edu, jandrews@ece.utexas.edu). T. Novlan and J. Zhang are with Samsung Research America-Dallas (email: t.novlan@samsung.com, jianzhong.z@samsung.com). Part of this paper was presented at IEEE Globecom 2015, $7^\text{th}$ International Workshop on Heterogeneous and Small Cell Networks (HetSNets)~\cite{li2015laawifiGC}.}}

\maketitle

\begin{abstract}
We leverage stochastic geometry to characterize key performance metrics for neighboring Wi-Fi and LTE networks in unlicensed spectrum. Our analysis focuses on a single unlicensed frequency band, where the locations for the Wi-Fi access points (APs) and LTE eNodeBs (eNBs) are modeled as two independent homogeneous Poisson point processes. Three LTE coexistence mechanisms are investigated: (1) LTE with continuous transmission and no protocol modifications; (2) LTE with discontinuous transmission; and (3) LTE with listen-before-talk (LBT) and random back-off (BO). For each scenario, we derive the medium access probability (MAP), the signal-to-interference-plus-noise ratio (SINR) coverage probability, the density of successful transmissions (DST), and the rate coverage probability for both Wi-Fi and LTE. Compared to the baseline scenario where one Wi-Fi network coexists with an additional Wi-Fi network, our results show that Wi-Fi performance is severely degraded when LTE transmits continuously. However, LTE is able to improve the DST and rate coverage probability of Wi-Fi while maintaining acceptable data rate performance when it adopts one or more of the following coexistence features: a shorter transmission duty cycle, lower channel access priority, or more sensitive clear channel assessment (CCA) thresholds.

\end{abstract}

\section{Introduction} As is well-established, licensed spectrum below 6 GHz is scarce and extremely expensive.   Given that there is over 400 MHz of generally lightly used unlicensed spectrum in the 5 GHz band -- e.g. in the USA, the U-NII bands from 5.15-5.35 GHz and 5.47-5.825 GHz ~\cite{FCC2013Revision} --  extending LTE's carrier aggregation capabilities to be able to opportunistically use such spectrum is an interesting proposition~\cite{qcom2013extend, ratasuk2012licenseexempt,3gpp2015TR}. Such an approach utilizes an anchor primary carrier in LTE operator's licensed spectrum holdings to provide control signaling and data, and a secondary carrier in the unlicensed spectrum that when available, offers a significant boost in data rate.   However, IEEE 802.11/Wi-Fi is an important incumbent system in these bands. Thus, a key design objective for LTE is to not only obey existing regulations for unlicensed spectrum, but also to achieve fair coexistence with Wi-Fi. In this paper, we propose a theoretical framework based on stochastic geometry~\cite{baccelli2010stochastic,baccelli2010stochasticpt2,stochtutorial,stochgeom,chiu2013stochastic} to analyze the coexistence issues that arise in such scenario. 

\subsection{Related Work and Motivation}
LTE is a centrally-scheduled system which was designed for exclusive usage of licensed spectrum. In contrast, Wi-Fi is built on distributed carrier sense multiple access with collision avoidance (CSMA/CA), where the carrier sensing mechanism allows transmissions only when the channel is sensed as idle. This distinctive medium access control (MAC) layer can potentially lead to very poor Wi-Fi performance when LTE operates in the same spectrum without any protocol modifications. Based on indoor office scenario simulations,~\cite{cavalcante2013performance,nihtila2013system} show that Wi-Fi is most often blocked by the LTE interference and that the throughput performance of Wi-Fi decreases significantly. In order to achieve fair coexistence with Wi-Fi, several modifications of LTE have been proposed. A simple approach which requires minimal changes to the current LTE protocol is to adopt a discontinuous transmission pattern, also known as LTE-U~\cite{qrc2014LTE,lteu2015TRv1}. By using the almost-blank subframes (ABS) feature to blank a certain fraction of LTE transmissions, Wi-Fi throughput can be effectively increased~\cite{nihtila2013system,almeida2013enabling,jeon2014lteworkshop}. This discontinuous transmission idea was previously adopted to address the coexistence issues of WiMax and Wi-Fi~\cite{kondo2010technology}.
Coexistence methodologies using the LBT feature, also known as licensed-assisted access (LAA) in 3GPP~\cite{3gpp2015TR}, have been considered in~\cite{jeon2014lteworkshop,mukherjee2015system}. 
In~\cite{jeon2014lteworkshop}, a random backoff mechanism with fixed contention window size is proposed in addition to LBT. The LAA operation of LTE in unlicensed spectrum is investigated in~\cite{mukherjee2015system}, which shows that the load-based LBT protocol of LAA with a backoff defer period can achieve fair coexistence. When LTE users adopts the LBT feature,~\cite{ratasuk2012licenseexempt} shows LTE can deliver significant uplink capacity even if it coexists with Wi-Fi.

All the aforementioned works are based on extensive system level simulations, which is usually very time-consuming due to the complicated dynamics of the overlaid LTE and Wi-Fi networks. Therefore, a mathematical approach would be helpful for more efficient performance evaluation and transparent comparisons of various techniques. A fluid network model is used in~\cite{jeon2014lte} to analyze the coexistence performance when LTE has no protocol modifications. However, the fluid network model is limited to the analysis of deterministic networks, which do not capture the multi-path fading effects and random backoff mechanism of Wi-Fi.
A centralized optimization framework is proposed in~\cite{sagari2015dynamic} to optimize the aggregate throughput of LTE and Wi-Fi. However, the analysis of~\cite{sagari2015dynamic} is based on Bianchi's model for CSMA/CA~\cite{bianchi2000performance}, which relies on the idealized assumption that the collision probability of the contending APs is ``constant and independent".

In recent years, stochastic geometry has become a popular and powerful mathematical tool to analyze cellular and Wi-Fi systems. Specifically, key performance metrics can be derived by modeling the locations of base stations (BSs)/access points (APs) as a realization of certain spatial random point processes. 
In~\cite{trac}, the coverage probability and average Shannon rate were derived for macro cellular networks with BSs distributed according to the complete spatial random Poisson point process (PPP). The analysis has been extended to several other cellular network scenarios, including heterogeneous cellular networks (HetNets)~\cite{Harpreet,dhillon2013load,mukherjee2012distribution}, MIMO~\cite{heath2013hetnets,dhillon2013downlink}, and carrier aggregation~\cite{lin2013carrieraggregation,zhang2014stochastic}. More realistic macro BS location models than PPP are investigated in~\cite{miyoshi2012cellular,li2015statistical,guo2014ADG}. Stochastic geometry can also model CSMA/CA-based Wi-Fi networks
. A modified Mat\'{e}rn hard-core point process, which gives a snapshot view of the simultaneous transmitting CSMA/CA nodes, has been proposed and validated in~\cite{nguyen2007dense80211} for dense 802.11 networks. This Mat\'{e}rn CSMA model is also used for analyzing other CSMA/CA based networks, such as ad-hoc networks with channel-aware CSMA/CA protocols~\cite{kim2014spatial}, and cognitive radio networks~\cite{nguyen2012stochasticCR}.

Due to its tractability for cellular and Wi-Fi networks, stochastic geometry is a natural candidate for analyzing LTE and Wi-Fi coexistence performance. In~\cite{bhorkar2014throughput}, the coverage and throughput performance of LTE and Wi-Fi were derived using stochastic geometry. However, the analytical Wi-Fi throughput in~\cite{bhorkar2014throughput} does not closely match the simulation results. Also, the effect of possible LTE coexistence methods, including discontinuous transmission and LBT with random backoff, were not investigated in~\cite{bhorkar2014throughput}. These shortcomings are addressed in this paper.


\subsection{Contributions}
In this work, a stochastic geometry framework is proposed to evaluate the coexistence performance of the neighboring Wi-Fi network and LTE network. 
Specifically, three coexistence scenarios are studied depending on the mechanism adopted by LTE, including: (1) LTE with continuous transmission and no protocol changes (i.e., conventional LTE); (2) LTE with fixed duty-cycling discontinuous transmission (i.e., LTE-U); and (3) LTE with LBT and random backoff mechanism (i.e., LAA). Several key performance metrics, including the MAP, the SINR coverage probability, the DST, and the rate coverage probability are derived under each scenario. The accuracy of the analytical results is validated against simulation results using SINR coverage probability. The main design insights of this paper can be summarized as follows:

(1) When LTE transmits continuously with no protocol changes, Wi-Fi performance is significantly impacted. Specifically, compared to the baseline scenario where Wi-Fi network coexists with an additional Wi-Fi network from another operator, the SINR coverage probability, DST, and rate coverage probability of Wi-Fi are severely degraded due to the persistent transmitting LTE eNBs. In contrast, LTE performance is shown to be relatively robust to Wi-Fi's presence.  

(2) When LTE transmits discontinuously with a fixed duty cycle, Wi-Fi generally has better DST and rate coverage under a synchronous muting pattern among LTE eNBs compared to the asynchronous one; and a short duty cycle for LTE transmission is required in both cases to protect Wi-Fi. Specifically, Wi-Fi achieves better performance under the synchronous case in general since it provides a much cleaner channel to Wi-Fi when LTE is muted. In contrast, since all eNBs transmit simultaneously under the synchronous case, LTE experiences stronger LTE interference and therefore worse DST and rate coverage compared to the asynchronous case.

(3) When LTE follows the LBT and random BO mechanism, LTE needs to accept either lower channel access priority or more sensitive CCA threshold to protect Wi-Fi. Specifically, Wi-Fi achieves better DST and rate coverage performance compared to the baseline scenario when LTE has either the same channel access priority (i.e., same contention window size) as Wi-Fi with more sensitive CCA threshold (e.g., -82 dBm), or lower channel access priority (i.e., larger contention window size) than Wi-Fi with less sensitive sensing threshold (e.g., -77 dBm). Under both scenarios, LTE is shown to maintain acceptable rate coverage performance. 




\section{System Model}~\label{SysModelSec}
In this section, we present the spatial location model for Wi-Fi APs and LTE eNBs, the radio propagation assumptions, and the channel access model for Wi-Fi and LTE.
\subsection{Spatial locations}
We focus on the scenario where two operators coexist in a single unlicensed frequency band with bandwidth $B$. Operator 1 uses Wi-Fi, while operator 2 uses LTE, which may implement certain coexistence methods to better coexist with operator 1. Both Wi-Fi and LTE are assumed to have full buffer downlink only traffic. The LTE eNBs are assumed to be low power small cell eNBs, such as femto-cell eNBs~\cite{chandrasekhar2008femtocell}. The locations for APs and eNBs are modeled as realizations of two independent homogeneous PPPs. Specifically, the AP process $\Phi_{W} = \{x_{i}\}_{i}$ has intensity $\lambda_W$\footnote{Note in any given time slot, not all Wi-Fi APs will be necessarily scheduled by CSMA/CA.}, while the eNB process $\Phi_{L} = \{y_{k}\}_{k}$ has intensity $\lambda_L$. Therefore, the number of APs and eNBs in any region with area $A$ are two independent Poisson random variables with mean $\lambda_W A$ and $\lambda_L A$ respectively (resp.). The PPP assumption for APs is reasonable due to the unplanned nature of most Wi-Fi deployments~\cite{nguyen2007dense80211}, while the PPP assumption for eNBs will exhibit similar SINR trend with a constant SINR gap compared to more accurate eNB location models~\cite{guo2014ADG}.

Both Wi-Fi stations (STAs) and LTE user equipments (UEs)\footnote{Wi-Fi STA and Wi-Fi users, as well as LTE UE and LTE users, are used interchangeably in this paper.} are also assumed to be distributed according to homogeneous PPPs. 
Each STA/UE is associated with its closest AP/eNB, which provides the strongest average received power. We assume the STA/UE intensity is much larger than the AP/eNB intensity, such that each AP/eNB has at least one STA/UE to serve. Since both STAs and UEs are homogeneous PPPs, we can analyze the performance of the typical STA/UE, which is assumed to be located at the origin. This is guaranteed by the independence assumption and Slyvniak's theorem\footnote{For any event $A$ and PPP $\Phi$, a heuristic interpretation of the Slyvniak's theorem is: $\mathbb{P} (\Phi \in A | o \in \Phi) = \mathbb{P} (\Phi \cup \{o\}\in A)$. }~\cite{chiu2013stochastic}.
Index $0$ is used for the serving AP/eNB to the typical STA/UE, which will be referred to as the closest or tagged AP/eNB for the rest of the paper. In addition, the link between the typical STA/UE and the tagged AP/eNB is referred to as the typical Wi-Fi/LTE link. Since $\Phi_W$ is a PPP with intensity $\lambda_W$, the probability density function (PDF) of the distance from the typical STA to the tagged AP is $f_{W}(r) = \lambda_W 2\pi r\exp(-\lambda_W  \pi r^2)$. Similarly, the PDF from the typical UE to the tagged eNB is $f_{L}(r) = \lambda_L 2\pi r\exp(-\lambda_L \pi r^2)$. 

\subsection{Propagation Assumptions}
 
The transmit power for each AP and eNB is assumed to be $P_W$ and respectively $P_L$. 
A common free space path loss model with reference distance of 1 meter is used for both Wi-Fi and LTE links, which is given by $\mathit{l}[\text{dB}](d) = 20 \log_{10} (\frac{4\pi}{\lambda_c}) + 10\alpha \log_{10} (d) $. Here $\lambda_c$ denotes the wavelength, $\alpha$ denotes the path loss exponent, and $d$ denotes the link length. The large-scale shadowing effects are neglected for simplicity. 
All the channels are assumed to be subject to i.i.d. Rayleigh fading, with each fading variable exponentially distributed with parameter $\mu$. The thermal noise power is $\sigma_{N}^{2}$. Notations and system parameters are listed in Table~\ref{SysParaTable}.
\begin{table}[t]
\center\caption{Notation and Simulation Parameters}\label{SysParaTable}
\resizebox{450pt}{!}{%
\begin{tabular}{|c|p{105mm}|p{45mm}|}
\hline 
Symbol & Definition & Simulation Value \\ 
\hline 
$\Phi_W$, $\lambda_W$ & Wi-Fi AP PPP and intensity & • \\ 
\hline 
$\Phi_L$, $\lambda_L$ & LTE eNB PPP and intensity & • \\ 
\hline 
 $P_W$, $P_L$ & Wi-Fi AP, LTE eNB transmit power & 23 dBm, 23 dBm \\ 
\hline 
$\Gamma_{cs}$, $\Gamma_{ed}$& Carrier sensing and energy detection thresholds & -82 dBm, -62 dBm \\ 
\hline 
$e_{i}^W$, $e_{k}^L$& Medium access indicator for AP $x_i$, eNB $y_k$ & \\ 
\hline
$x_0$, $y_0$ & The tagged AP and tagged eNB (i.e., the AP and eNB closest to the typical STA and UE resp.)& \\
\hline
$f_{W}(r), f_{L}(r)$ & PDF of the distance from tagged AP/eNB to typical STA/UE &\\
\hline 
$f_c$, $B$ & Carrier frequency and bandwidth of the unlicensed band & 5 GHz, 20 MHz \\ 
\hline 
$\alpha$ & Path loss exponent & 4 \\ 
\hline 
$\mu$ & Parameter for Rayleigh fading channel & 1 \\ 
\hline 
$\sigma_{N}^{2}$ & Noise power & 0 \\ 
\hline 
$B(x,r)$ ($B^{o}(x,r)$) & Closed (open) ball with center $x$ and radius $r$ & \\
\hline
$B^{c}(x,r)$  &  Complement of $B(x,r)$ &\\
\hline
$F_{i,0}^{L}$ ($F_{i,0}^{W}$, $F_{i,0}^{LW}$, $F_{i,0}^{WL}$) & Fading of the channel from eNB $y_{i}$ to typical UE (AP $x_{i}$ to typical STA, eNB $y_{i}$ to typical STA, AP $x_{i}$ to typical UE)& exponentially distributed with parameter $\mu$\\
\hline
$G_{i,j}^{L}$ ($G_{i,j}^{W}$, $G_{i,j}^{LW}$, $G_{i,j}^{WL}$) & Fading of the channel from eNB $y_{i}$ to eNB $y_{j}$ (AP $x_{i}$ to AP $x_{j}$, eNB $y_{i}$ to AP $x_{j}$, AP $x_{i}$ to eNB $y_{j}$) & exponentially distributed with parameter $\mu$\\
\hline
\end{tabular} }
\end{table}
\subsection{Modeling Channel Access for Wi-Fi}\label{MACDefnSubsec}
In contrast to LTE, 
Wi-Fi implements the distributed CSMA/CA protocol for channel access coordination among multiple APs. The CSMA/CA protocol consists of the physical layer clear channel assessment (CCA) process and a random backoff mechanism, such that two nearby nodes will never transmit simultaneously. 
In particular, the Wi-Fi device will hold CCA as busy if
any valid Wi-Fi signal that exceeds the carrier sense (CS) threshold $\Gamma_{cs}$ is detected,
or if any signal that exceeds the energy detection threshold (ED) $\Gamma_{ed}$ is received~\cite{ieee2012IEEE}. Similar to~\cite{jeon2014lte}, we assume Wi-Fi devices detect the eNB transmission with the energy detection threshold $\Gamma_{ed}$ since an LTE signal is not decodable. As soon as a CSMA/CA device observes an idle channel, it needs to follow a random back-off period before transmission. This back-off period is chosen randomly from a set of possible values called the contention window.

To model the locations of Wi-Fi APs which simultaneously access the channel at a given time, we adapt the formulation of~\cite{nguyen2007dense80211} to account for the coexisting LTE network. 
We can define the contender of a Wi-Fi AP $x_{i}$ as the other Wi-Fi APs and the LTE eNBs whose power received by $x_{i}$ exceeds the threshold $\Gamma_{cs}$ and $\Gamma_{ed}$ respectively. Each Wi-Fi AP $x_{i}$ has an independent mark $t_{i}^W$ to represent the random back-off period, which is uniformly distributed on $[0,1]$. Each Wi-Fi AP obtains channel access for packet transmission if it chooses a smaller timer, i.e., back-off period, than all its contenders. A medium access indicator $e_i^{W}$ is assigned to each AP, which is equal to 1 if the AP is allowed to transmit by the CSMA/CA protocol, and 0 otherwise. Depending on the specific coexistence mechanism of LTE, the medium access indicator for each AP is determined differently. The Palm probability~\cite[p.131]{chiu2013stochastic} that the medium access indicator of a Wi-Fi AP is equal to 1 is referred to as the medium access probability, or MAP for short.

The considered channel access mechanism has some limitations, such as it has a fixed contention window size which does not capture the exponential backoff, and it is also more suitable for synchronized and slotted version of CSMA/CA. Nevertheless, it is able to model the key feature of CSMA/CA in IEEE 802.11 standard~\cite{ieee2012IEEE}, such that each CSMA/CA device transmits if it does not carrier sense any other CSMA/CA device with a smaller back-off timer. In addition, through comparisons with simulation results,~\cite{nguyen2007dense80211,busson_inria_00316029} show this simplified model provides a reasonable conservative representation of transmitting APs in the actual CSMA/CA networks. 
\subsection{Definition of Performance Metrics}\label{SINRDefnSubsec}
The main performance metrics that are analyzed include the MAP of the tagged AP and eNB, as well as the SINR coverage probability for the typical Wi-Fi STA and LTE UE. 
Specifically, given the tagged AP $x_0$ transmits (i.e., $e_{0}^{W} = 1$), the received SINR of the typical Wi-Fi STA is:
\begin{align}\label{SINRWiFiW1Eq}
\allowdisplaybreaks
	&\text{SINR}_{0}^{W} = \frac{P_W F_{0,0}^{W} / \mathit{l}(\|x_{0}\|)}{\sum\limits_{x_{j} \in \Phi_W \setminus \{ x_{0}\}} \!\!\!\!\!\!\!\!\!\! P_W F_{j,0}^{W} e_{j}^W / \mathit{l}(\|x_{j}\|)+ \sum\limits_{y_{m} \in \Phi_{L}} \!\!\!\!P_L   F_{m,0}^{LW} e_{m}^L/ \mathit{l}(\|y_{m}\|) + \sigma_{N}^2},
\end{align}
where $e_{j}^W$ and $e_{m}^L$ represent the medium access indicator for AP $x_j$ and eNB $y_m$ respectively.
The SINR coverage probability of the typical STA with SINR threshold $T$ is defined as $\mathbb{P}(\text{SINR}_{0}^{W} > T | e_{0}^{W} = 1)$, which gives the instantaneous SINR performance of the typical Wi-Fi link. 
Similarly, the received SINR of the typical LTE UE given the tagged eNB $y_0$ transmits is:
\begin{align}\label{SINRLTEEq}
\allowdisplaybreaks
&\text{SINR}_{0}^{L}  = \frac{P_L F_{0,0}^{L}/ \mathit{l}(\|y_{0}\|)}{\sum\limits_{x_{j} \in \Phi_{W}} \!\!\!\!P_W  F_{j,0}^{WL} e_{j}^W/ \mathit{l}(\|x_{j}\|) + \sum\limits_{y_{m} \in \Phi_L \setminus \{ y_{0}\}}  \!\!\!\!\!\!\!\!\!\! P_L   F_{m,0}^{L} e_{m}^L/ \mathit{l}(\|y_{j}\|)+ \sigma_{N}^2},
\end{align}and the SINR coverage probability is $\mathbb{P}(\text{SINR}_{0}^{L} > T | e_{0}^L = 1)$.

Based on the MAP and the SINR distribution, we will compare different LTE coexistence mechanisms using the density of successful transmission and the rate coverage probability, which are defined as follows.
\begin{definition}[Density of Successful Transmissions]\label{DSTDefn}
	For decoding SINR requirement $T$, the density of successful transmission, or DST for short, is defined as the mean number of successful transmission links per unit area~\cite{baccelli2010stochasticpt2}. Since the typical Wi-Fi/LTE link is activated only when the tagged AP/eNB accesses the channel, the DST for Wi-Fi and LTE are given by:
	\allowdisplaybreaks
	\begin{align}
	d_{suc}^W(\lambda_W,\lambda_L,T) &= \lambda_W \mathbb{E}[e_{0}^{W}] \mathbb{P}(\text{SINR}_{0}^{W} > T | e_{0}^W = 1),\nonumber\\
	d_{suc}^L(\lambda_W,\lambda_L,T) &= \lambda_L \mathbb{E}[e_{0}^{L}] \mathbb{P}(\text{SINR}_{0}^{L} > T | e_{0}^L = 1).\label{DSTDefnEq}
	\end{align}
\end{definition}
\begin{definition}[Rate coverage]\label{RCOPDefn}
The rate coverage probability with threshold $\rho$ is defined as the probability for tagged Wi-Fi AP/LTE eNB to support an aggregate data rate of $\rho$, given by\footnote{The user-perceived data rate distribution can be obtained from~(\ref{RCOPDefnEq}) by considering the average fraction of resource that each user achieves.}:
	\allowdisplaybreaks
		\begin{align}
		P_{rate}^W(\lambda_W,\lambda_L,\rho) &=  \mathbb{P}(B \log(1+\text{SINR}_{0}^{W}) \mathbb{E}[e_{0}^{W}] > \rho | e_{0}^W = 1),\nonumber\\
		P_{rate}^L(\lambda_W,\lambda_L,\rho) &=  \mathbb{P}(B \log(1+\text{SINR}_{0}^{L}) \mathbb{E}[e_{0}^{L}] > \rho | e_{0}^L = 1).\label{RCOPDefnEq}
		\end{align}
\end{definition}The $\mathbb{E}[e_{0}^{W}]$ and $\mathbb{E}[e_{0}^{L}]$ in~(\ref{RCOPDefnEq}) accounts for the fact that the tagged AP and tagged eNB have channel access for $\mathbb{E}[e_{0}^{W}]$ and $\mathbb{E}[e_{0}^{L}]$ fraction of time respectively. 
Equivalently, the rate coverage probability gives the fraction of Wi-Fi APs/LTE eNBs (or Wi-Fi/LTE cells) that can support an aggregate data rate of $\rho$ for the rest of the paper.
\begin{remark}\label{RemarkAngleInvariance}
Since both $\Phi_W$ and $\Phi_L$ are stationary and isotropic, the above performance metrics are invariant with respect to (w.r.t.) the angle of the tagged AP $x_0$ and tagged BS $y_0$. Without loss of generality, the angle of $x_0$ and $y_0$ are assumed to be $0$. In addition, the PDF of $\|x_0\|$ and $\|y_0\|$ are given by $f_W(\cdot)$ and $f_L(\cdot)$ respectively, which are defined in Table~\ref{SysParaTable}. 
\end{remark}

Finally, we define several functions that will be used throughout this paper in Table II. Specifically, $N^L_0(y,r,\Gamma)$ and $N^W_0(y,r,\Gamma)$ represent the expected number of eNBs and APs respectively in $\mathbb{R}^2 \setminus B(0,r)$, whose signal power received at $y \in \mathbb{R}^2$ exceeds $\Gamma$. In addition, $C^L_0(y_1, \Gamma_1, y_2, \Gamma_2)$ and $C^W_0(y_1, \Gamma_1, y_2, \Gamma_2)$ represent the expected number of eNBs and APs respectively in $\mathbb{R}^2 \setminus B(0,\|y_2\|)$, whose signal powers received at $y_1 \in \mathbb{R}^2$ and $y_2 \in \mathbb{R}^2$ exceed $\Gamma_1$ and $\Gamma_2$ respectively. Moreover, $M$, $V$ and $U$ are functions helping to calculate the conditional MAP in the following sections.

\begin{table}[h]
	\renewcommand{\arraystretch}{1}
	\centering\caption{Notations and Definitions of Special Functions}\label{FunctionTable}
	\begin{tabular}{|M{40mm}|M{120mm}|} \hline
		Notation & Definition \\ \hline
		\small $N^L_{0}(y,r, \Gamma)$ & \small $\lambda_L \int_{\mathbb{R}^2\setminus B(0,r)} \exp(-\mu \frac{\Gamma}{P_L} \mathit{l}(\|x-y\|)) {\rm d}x$ \\ \hline
		\small $N^W_0(y,r, \Gamma)$ & \small $\lambda_W \int_{\mathbb{R}^2\setminus B(0,r)} \exp(-\mu \frac{\Gamma}{P_W} \mathit{l}(\|x-y\|)) {\rm d}x$  \\ \hline 
		\small $N^L_1(r, \Gamma)$, \small $N^W_1(r, \Gamma)$ & \small $N^L_0(y,r, \Gamma)$, \small $N^W_0(y,r, \Gamma)$ (polar coordinates of $y = (r,0)$) \\ \hline
		\small $N^L_2(r)$, \small $N^W_2(r)$ &  \small $N^L_0(y, r, \Gamma_{ed})$, \small $N^W_0(y, r, \Gamma_{cs})$ (polar coordinates of $y = (r,0)$)\\ \hline 
		\small $N^L_3(\Gamma)$, \small $N^W_3(\Gamma)$ & \small $N^L_0(o,0,\Gamma)$, \small $N^W_0(o,0,\Gamma)$ \\ \hline	
		\small $N^L$, \small $N^W$ & \small $N^L_0(o,0,\Gamma_{ed})$, \small $N^W_0(o,0,\Gamma_{cs})$ \\ \hline
		\small $C^L_{0}(y_1, \Gamma_1, y_2, \Gamma_2)$ & \small $\lambda_L \int_{\mathbb{R}^2 \setminus B(0,\|y_2\|)} \exp(-\mu \frac{\Gamma_{1}}{P_L} \mathit{l}(\|x-y_1\|) - \mu \frac{\Gamma_{2}}{P_L} l(\|x-y_2\|)) {\rm d}x$  \\ \hline
		\small $C^W_0(y_1, \Gamma_1, y_2,\Gamma_2) $& \small $ \lambda_W \int_{\mathbb{R}^2 \setminus B(0,\|y_2\|) } \exp(-\mu \frac{\Gamma_{1}}{P_W} \mathit{l}(\|x-y_1\|) - \mu \frac{\Gamma_{2}}{P_W} l(\|x-y_2\|)) {\rm d}x$ \\ \hline 
		\small $C^L_1(y_1,y_2)$, \small $C^W_1(y_1,y_2)$ & \small $C^L_0(y_1,\Gamma_{ed},y_2,\Gamma_{ed})$, \small $C^W_0(y_1,\Gamma_{cs},y_2,\Gamma_{cs})$ \\ \hline
		\small $C^L_2(y_1)$, \small $C^W_2(y_1)$ & \small $C^L_0(y_1,\Gamma_{ed},o,\Gamma_{ed})$, \small $C^W_0(y_1,,\Gamma_{cs},o,\Gamma_{cs})$ \\ \hline
		\small $M(N_1,N_2,N_3)$  & $ \small (\frac{1-\exp(-N_1)}{N_1}  - \frac{1-\exp(-N_1-N_2+N_3)}{N_1+N_2-N_3})/(N_2-N_3) $\\ \hline
		\small $V(x,s_1,s_2,N_1,N_2,N_3)$ &  \small$(1-\exp(-\mu s_1\mathit{l}(\|x\|))) M(N_1,N_2,N_3)+(1-\exp(-\mu s_2\mathit{l}(\|x\|))) M(N_2,N_1,N_3)$ \\ \hline
		\small $U(x,s,N_1)$  & \small $\frac{1-\exp(-N_1)}{N_1} - \exp(-\mu s \mathit{l}(\|x\|) )(\frac{1-\exp(-N_1)}{N_1^2} - \frac{\exp(-N_1)}{N_1})$ \\ \hline
	\end{tabular}
\end{table}

\section{LTE with Continuous Transmission and No Protocol Change}\label{LTEwCT}
In this section, the MAP and SINR coverage performance for the LTE and Wi-Fi networks are investigated when LTE transmits continuously without any protocol modifications.
\subsection{Medium Access Probability}
From the CSMA/CA protocol described in Section II-C, a Wi-Fi AP will not transmit whenever it has an LTE eNB as its contender, i.e., the power it receives from any LTE eNB exceeds the energy detection threshold $\Gamma_{ed}$. Therefore, the medium access indicator $e_{i}^W$ for AP $x_{i}$ is: 
\begin{align}\label{MAIWi-FiW1}
e_{i}^W = & \prod\limits_{y_{k} \in \Phi_{L}} \mathbbm{1}_{G_{ki}^{LW}/ \mathit{l}(\|y_{k} - x_{i}\|) \leq \frac{\Gamma_{ed}}{P_L}}  \prod\limits_{x_{j} \in \Phi_{W} \setminus \{x_{i}\}} \biggl(\mathbbm{1} _{t_{j}^{W} \geq t_{i}^{W}}  +\mathbbm{1} _{t_{j}^{W} < t_{i}^{W}} \mathbbm{1}_{G_{ji}^{W}/ \mathit{l}(\|x_{j}- x_{i}\|) \leq \frac{\Gamma_{cs}}{P_W }} \biggl) .
\end{align}
The first part of~(\ref{MAIWi-FiW1}) means each Wi-Fi AP will not transmit whenever it has any LTE contender, while the second part of~(\ref{MAIWi-FiW1}) means each Wi-Fi AP will not transmit whenever any of its Wi-Fi contenders has a smaller back-off timer. Although~(\ref{MAIWi-FiW1}) is consistent with IEEE 802.11 specifications~\cite{ieee2012IEEE}
, energy detection is typically implemented based on total interference~\cite{3gpp2015TR}, i.e., each AP will report channel as busy if the total (non Wi-Fi) interference exceeds the energy detection threshold $\Gamma_{ed}$.
Nevertheless, under the assumption that eNBs/APs have a PPP distribution,~(\ref{MAIWi-FiW1}) is a reasonable model for the total interference based energy detection since: (1) the tail distribution of the total interference asymptotically approaches that of the interference from the strongest interferer~\cite{haenggi2009interference}; (2) ED threshold $\Gamma_{ed}$ is 20 dB higher than the CS threshold $\Gamma_{cs}$, which makes $\Gamma_{ed}$ a relatively large number; 
(3) simulation results in Fig.~\ref{RSSIvsStrongestFig} show that given $\Gamma_{ed} = -62$dBm and AP $x_i$, $\mathbb{P}(\sum_{y_{k} \in \Phi_{L}} \frac{ P_L G_{ki}^{LW}}{ \mathit{l}(\|y_{k} - x_{i}\|)} \leq \Gamma_{ed}) \approx \mathbb{P}(\max_{y_{k} \in \Phi_{L}} \frac{P_L G_{ki}^{LW}}{\mathit{l}(\|y_{k} - x_{i}\|) }\leq \Gamma_{ed})$ for various values of $\alpha$ and $\lambda_L$; and (4) there is no known closed-form interference distribution with PPP distributed transmitters~\cite{haenggi2009interference}.
\begin{figure}
	\begin{subfigure}[b]{0.45\textwidth}
		\centering
		\includegraphics[height=2.0in, width= 3in]{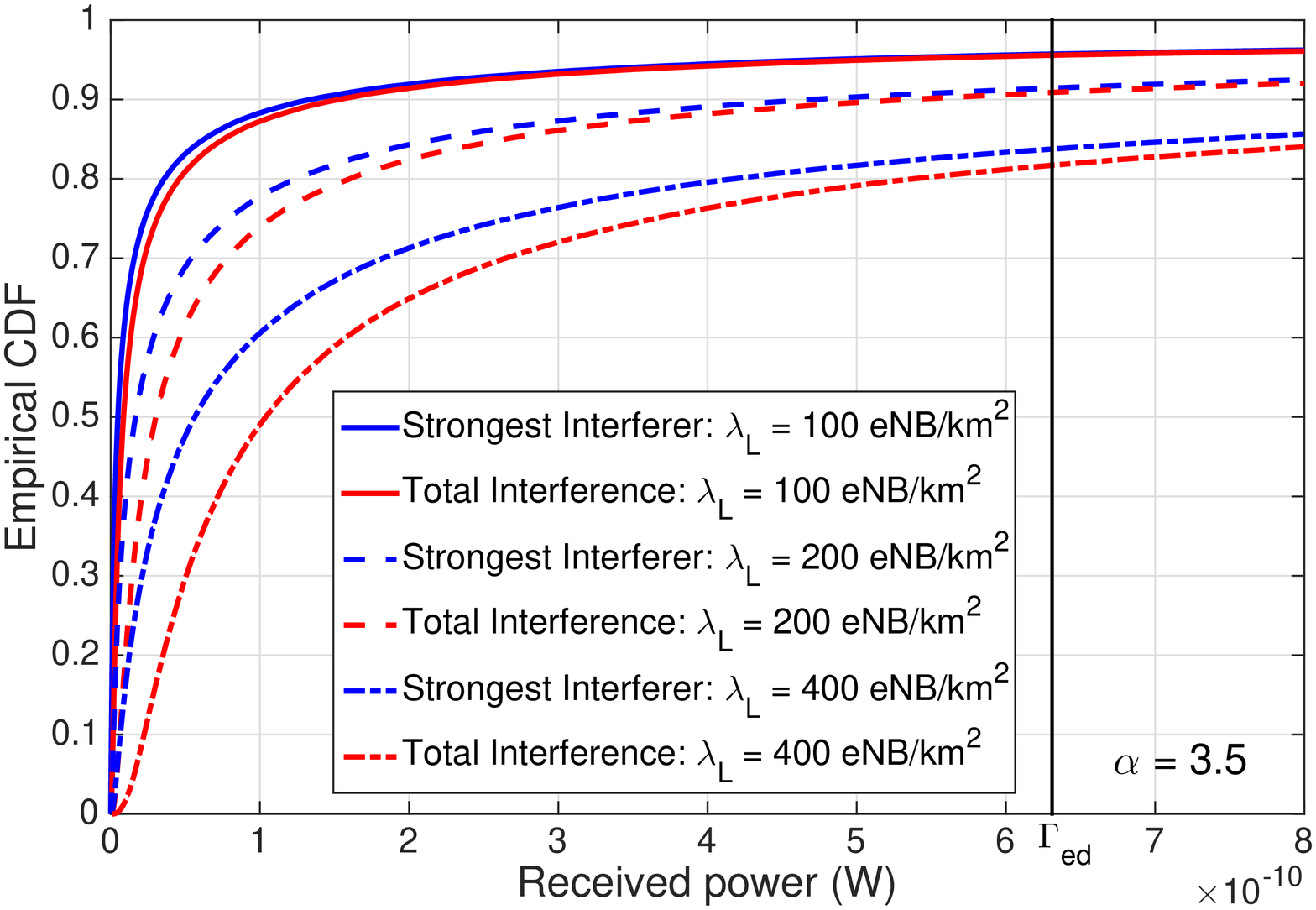}
	\end{subfigure}
	\hfill
	\begin{subfigure}[b]{0.45\textwidth}
		\centering
		\includegraphics[height=2.0in, width=3in]{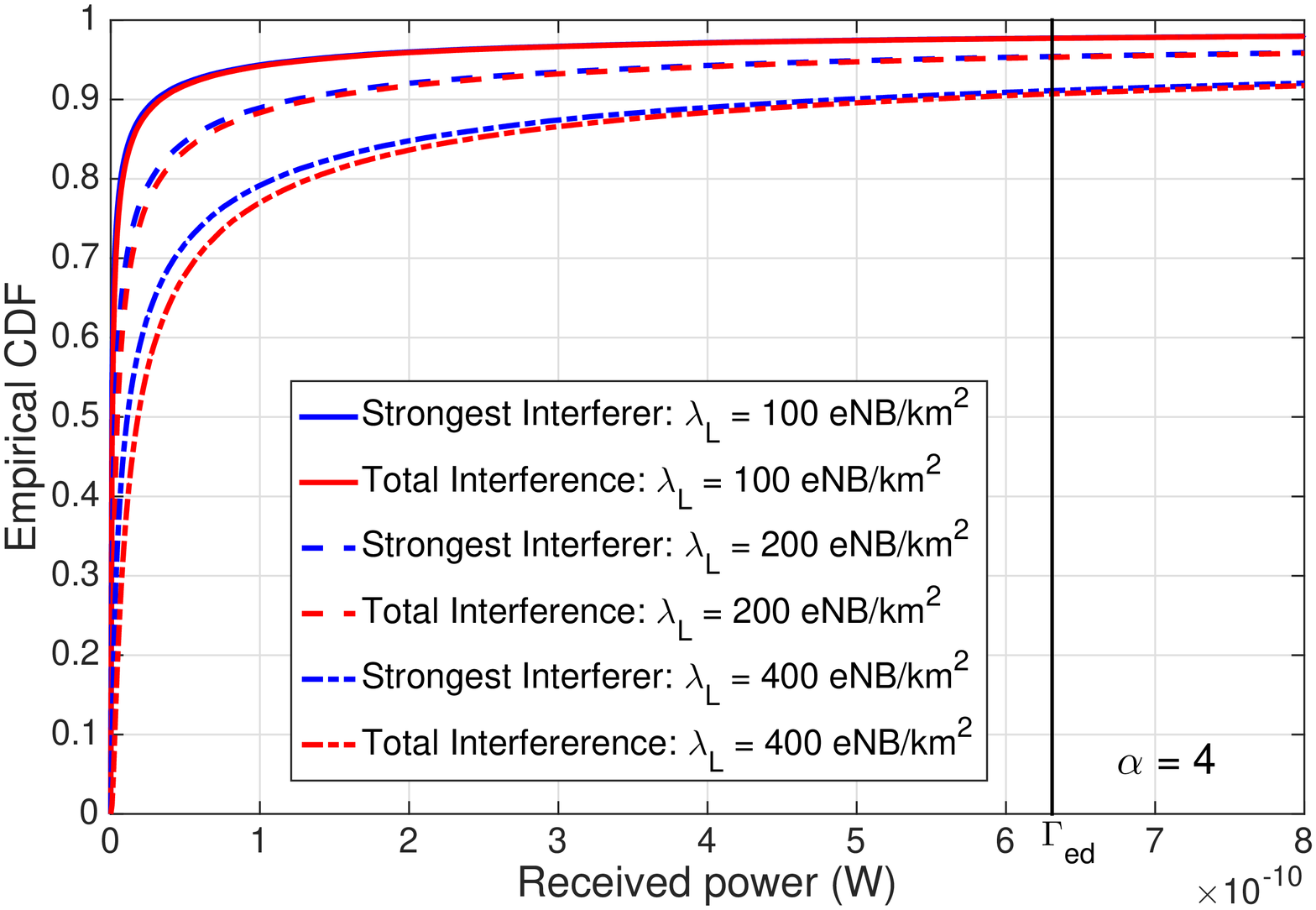}
	\end{subfigure}
	\caption{Empirical CDF of total LTE interference and strongest LTE interferer at a typical AP.}\label{RSSIvsStrongestFig}
\end{figure}
\begin{lemma}\label{MAPWi-FiW1lemma}
When LTE transmits continuously with no protocol modifications, the MAP for a typical Wi-Fi AP is given by:
\allowdisplaybreaks
\begin{align}\label{MAPWi-FiW1Eq}
\allowdisplaybreaks
p_{0,\text{MAP}}^W(\lambda_W,\lambda_L) = \exp(-N^{L}) \frac{1-\exp(-N^{W})}{N^{W}},
\end{align}where $N^{L}$ and $N^{W}$ 
are defined in Table~\ref{FunctionTable}.
\end{lemma}
\begin{proof}
The MAP of Wi-Fi AP $x_i$ is the Palm probability that its medium access indicator is equal to 1. Given its timer $t_i^{W} = t$, the MAP can be derived as:
\allowdisplaybreaks
\begin{align*}
\allowdisplaybreaks
&\mathbb{E}_{\Phi_{W}}^{x_i} \biggl[ \prod\limits_{y_{k} \in \Phi_{L}} \mathbbm{1}_{G_{ki}^{LW}/ \mathit{l}(\|y_{k} - x_{i}\|) \leq \frac{\Gamma_{ed}}{P_L}}  \prod\limits_{x_{j} \in \Phi_{W} \setminus \{x_{i}\}} \biggl(\mathbbm{1} _{t_{j}^{W} \geq t}  +\mathbbm{1} _{t_{j}^{W} < t} \mathbbm{1}_{G_{ji}^{W}/ \mathit{l}(\|x_{j}- x_{i}\|) \leq \frac{\Gamma_{cs}}{P_W }} \biggl)\biggl] \\
\overset{(a)}{=} &\mathbb{E}\left[\prod_{k} \left(1-\exp(-\mu \frac{\Gamma_{ed}}{P_L} \mathit{l}(\|y_{k}-x_i\|))\right)\right]  \times \mathbb{E}_{\Phi_{W}}^{!x_i} \biggl[\prod_{j} (1-t\exp(-\mu \frac{\Gamma_{cs}}{P_W} \mathit{l}(\|x_j-x_i\|)) \biggl]\\
\overset{(b)}{=} &\exp(-N^{L})\exp(-t N^{W}),
\end{align*}where (a) follows from the fact that $\Phi_L$ is independent of $\Phi_W$, and (b) follows from Slyvniak's theorem and the PGFL of a homogeneous PPP. 
Finally noting that $t \sim U(0,1)$ and deconditioning on $t$ gives the desired result.  
\end{proof}

\begin{remark}
By adding the LTE network with intensity $\lambda_L$, the MAP for a typical AP is degraded by $\exp(-N^{L})$ compared to the Wi-Fi only scenario. Note that the decrease is exponential w.r.t. $\lambda_L$, the LTE eNB intensity. 
\end{remark}

Based on the system parameters listed in Table~\ref{SysParaTable}, the MAP for the typical Wi-Fi AP is plotted in Fig.~\ref{MAPW1Wi-FiFig} w.r.t. different AP and eNB intensities. From Fig.~\ref{MAPW1Wi-FiFig}, it can be observed that with low LTE eNB intensity (e.g. LTE eNB intensity is less than 100/km$^2$), the MAP for the typical Wi-Fi AP is not much affected by the additional LTE network as a result of the high energy detection threshold for LTE signals. However, when the LTE eNB intensity increases to over 100/km$^2$, the additional eNBs significantly degrade the MAP of the typical Wi-Fi AP. 
\begin{figure}[h]
		\centering
 		\includegraphics[height=2in, width=3.2in]{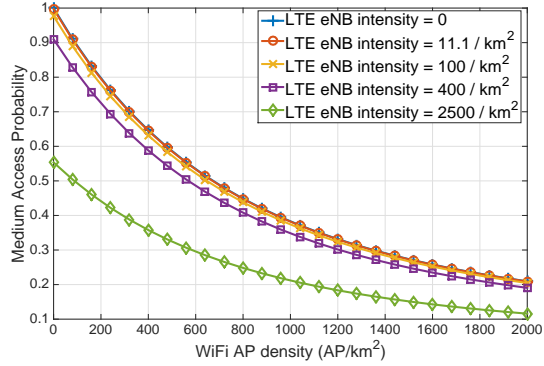}
        \caption{Effect of AP and eNB intensities on the MAP for the typical Wi-Fi AP.}\label{MAPW1Wi-FiFig}
\end{figure}

Since the tagged Wi-Fi AP is closer to the typical STA than other APs, the MAP of the tagged AP will be a biased version for the MAP of the typical AP:
\begin{corollary}\label{MAPWi-FitagW1lemma}
	When LTE transmits continuously with no protocol modifications, the MAP for the tagged Wi-Fi AP is given by:
	\allowdisplaybreaks
	\begin{align}\label{MAPWi-FitageW1Eq}
	\allowdisplaybreaks
	\hat{p}_{0,\text{MAP}}^W(\lambda_W,\lambda_L) =  \int_{0}^{\infty}\frac{1-\exp(-N^W_2(r_0))}{N^W_2(r_0)} \exp(-N^L) f_{W}(r_0){\rm d}r_0,
	\end{align}
	where $f_W$ is defined in Table~\ref{SysParaTable}, while $N^L$ and $N^W_2$ are defined in Table~\ref{FunctionTable}.
\end{corollary}
\begin{proof}
According to Remark~\ref{RemarkAngleInvariance}, given the tagged AP is located at $x_0 = (r_0,0)$, we have:
\allowdisplaybreaks
\begin{align*}
\allowdisplaybreaks
&\mathbb{P}(e_0^W = 1 | x_0 = (r_0,0)) \\
= &\mathbb{E}_{\Phi_W}^{x_0}\left(\prod\limits_{y_{k} \in \Phi_{L}} \mathbbm{1}_{G_{ki}^{LW}/ \mathit{l}(\|y_{k} - x_{0}\|) \leq \frac{\Gamma_{ed}}{P_L}} \prod\limits_{x_{j} \in \Phi_{W} \setminus \{x_{0}\}} \left(\mathbbm{1} _{t_{j}^{W} \geq t_{0}^{W}}  +\mathbbm{1} _{t_{j}^{W} < t_{0}^{W}} \mathbbm{1}_{G_{j0}^{W}/ \mathit{l}(\|x_{j}- x_{0}\|) \leq \frac{\Gamma_{cs}}{P_W }} \right)| \Phi_W(B^o(0,r_0)) = 0 \right)\\
=& \mathbb{E}\left(\prod\limits_{y_{k} \in \Phi_{L}} \mathbbm{1}_{G_{k0}^{LW}/ \mathit{l}(\|y_{k} - x_{0}\|) \leq \frac{\Gamma_{ed}}{P_L}} \prod\limits_{x_{j} \in \Phi_{W}\cap B^c(0,r_0)} \left( \mathbbm{1} _{t_{j}^{W} \geq t_{0}^{W}}  +\mathbbm{1} _{t_{j}^{W} < t_{0}^{W}} \mathbbm{1}_{G_{j0}^{W}/ \mathit{l}(\|x_{j}- x_{0}\|) \leq \frac{\Gamma_{cs}}{P_W }} \right)\right)\\
=&\frac{1-\exp(-N^W_2(r_0))}{N_2^W(r_0)} \exp(-N^L),
\end{align*}Finally, Corollary~\ref{MAPWi-FitagW1lemma} is derived by incorporating the distribution of $\|x_0\|$. 
\end{proof}
\subsection{SINR Coverage Probability}
\subsubsection{SINR Coverage Probability of Typical Wi-Fi STA}
Since LTE eNBs transmit continuously with no protocol modifications, the medium access indicator for each LTE eNB is 1 almost surely. The medium access indicator $e_j^W$ in~(\ref{MAIWi-FiW1}) depends on both $\Phi_L$ and $\Phi_W$. So there exists a correlation between the interference from LTE eNBs and that from the Wi-Fi APs. Later we will show that if we substitute $\Phi_L$ by another independent PPP $\Phi_L^{'}$ with intensity $\lambda_L$ in~(\ref{MAIWi-FiW1}), the corresponding SINR coverage is an accurate approximation. This means the correlation between the interference from eNBs and APs is mostly captured by the statistical effect of $\Phi_L$ on determining the MAP for Wi-Fi APs. Given the tagged AP is located at $x_0$, we first derive the conditional MAP for another Wi-Fi AP and $x_0$ to transmit simultaneously. 


\begin{corollary}\label{PWSelecProbW1Wi-FiCoro}
	Conditionally on the fact that the tagged AP $x_0 = (r_0,0)$ transmits, the probability for another AP $x \in \Phi_W \cap B^c(0,r_0)$ to transmit is:
	\allowdisplaybreaks
	\allowdisplaybreaks
	\begin{align}
	\allowdisplaybreaks
	h_{1}(r_0,x) = \frac{V(x-x_0,\frac{\Gamma_{cs}}{P_W},\frac{\Gamma_{cs}}{P_W},N^W_2(r_0),N^W_0(x,r_0,\Gamma_{cs}),C^W_1(x,x_0)) }
	{U(x-x_0,\frac{\Gamma_{cs}}{P_W},N^W_2(r_0))\exp(N^L - C^L_2(x-x_0))},\label{PWSelecProbW1Wi-FiEq}
	\end{align}where $B^c(0,r_0)$ is defined in Table~\ref{SysParaTable}.
\end{corollary}

The proof of Corollary~\ref{PWSelecProbW1Wi-FiCoro} is provided in Appendix~\ref{CondiCOPWiFiAppdx}. 
Then the SINR coverage performance of the typical STA, denoted by $p_0^W(T, \lambda_W,\lambda_L)$, is obtained as follows:

\begin{lemma}\label{COPW1Wi-FiLemma}
	The SINR coverage probability of the typical Wi-Fi STA with the SINR threshold $T$ can be approximated as:
	\begin{align}\label{COPCondiW1Wi-FiEq}
	\allowdisplaybreaks
	p_0^W(T,\lambda_W,\lambda_L)  \approx &\int_{0}^{\infty} \exp\biggl( -\mu T \mathit{l}(r_0) \frac{\sigma_{N}^2}{P_W}\biggr) \exp\biggl(  -\int_{\mathbb{R}^2} \frac{T \mathit{l}(r_0) \lambda_L}{\frac{P_W}{P_L}\mathit{l}(\|x\|) + T \mathit{l}(r_0)} {\rm d}x\biggr) \nonumber \\
	&\times \exp\biggl(- \int_{\mathbb{R}^2 \setminus B(0,r_0)} \frac{T\mathit{l}(r_0) \lambda_W h_1(r_0,x)}{\mathit{l}(\|x\|) + T \mathit{l}(r_0)}   {\rm d}x \biggl) f_{W}(r_0) {\rm d}r_0.
	\end{align} 
\end{lemma} 

\begin{proof}
	The conditional SINR coverage of the typical Wi-Fi STA is derived as follows:
	\allowdisplaybreaks	
	\begin{align*}
	&\mathbb{P}(\text{SINR}_{0}^{W} > T | x_0 = (r_0,0),e_0^W = 1) \nonumber\\
	\overset{(a)}{=}& \mathbb{P}_{\Phi_W}^{x_0}(\frac{ F_{0,0}^{W} / \mathit{l}(\|x_{0} \|)}{  \sum\limits_{x_{j} \in \Phi_W \setminus  \{ x_{0}\}}  F_{j,0}^{W} e_{j}^W / \mathit{l}(\|x_{j} \|)+ \sum\limits_{y_{m} \in \Phi_{L}}\frac{P_L}{P_W}   F_{m,0}^{LW}/ \mathit{l}(\|y_{m}\|) + \frac{\sigma_{N}^2}{P_W}} > T | \Phi_W(B^o(0,r_0))= 0, e_0^W = 1)\nonumber\\
	\overset{(b)}{=}& \mathbb{P}(\frac{ F_{0,0}^{W} / \mathit{l}(\|x_{0} \|)}{  \sum\limits_{x_{j} \in \Phi_W \cap B^c(0,r_0)}  F_{j,0}^{W} \hat{e}_{j}^W / \mathit{l}(\|x_{j} \|)+ \sum\limits_{y_{m} \in \Phi_{L}}\frac{P_L}{P_W}   F_{m,0}^{LW}/ \mathit{l}(\|y_{m}\|) + \frac{\sigma_{N}^2}{P_W}} > T | \hat{e}_{0}^W = 1)\nonumber\\
	\overset{(c)}{\approx} &\exp( -\mu T \mathit{l}(r_0) \frac{\sigma_{N}^2}{P_W}) \mathbb{E}\biggl[-\mu T \mathit{l}(r_0)( \!\!\!\!\!\!\!\sum\limits_{x_{i} \in \Phi_{W}\cap B^c(0,r_0)  } \!\!\!\!\!\!\frac{P_W}{P_L}  \frac{F_{i,0}^{WL} \hat{e}_{i}^W }{\mathit{l}(\|x_{i}\|)}) \biggl| \hat{e}_{0}^W = 1 \biggl]\mathbb{E}\biggl[-\mu T \mathit{l}(r_0)(\sum\limits_{y_{m} \in \Phi_L }\frac{ F_{m,0}^{L}}{\mathit{l}(\|y_{m}\|)}) \biggl],
	\end{align*}where (a) follows from Baye's rule by re-writing $x_0 = (r_0,0)$ as $x_0 \in \Phi_W$ and $\Phi_W(B^o(0,r_0))= 0$. Here $B^o(0,r_0)$ is defined in Table~\ref{SysParaTable}. Step (b) is derived from Slyvniak's theorem and by de-conditioning on $\Phi_W(B^o(0,r_0))= 0$. The modified medium access indicator for AP $x_i \in (\Phi_{W}\cap B^c(0,r_0)+\delta_{x_0}) $ is given by~(\ref{MAIW1Wi-FiEq}).
	The conditional probability for the Wi-Fi AP $x_j \in \Phi_W \cap B^c(0,r_0)$ to transmit given $x_0$ transmits, i.e., $\mathbb{P}(\hat{e}_i^W  = 1 | \hat{e}_0^W =1)$, is derived in Corollary~\ref{PWSelecProbW1Wi-FiCoro}. Step (c) uses the assumption that the interference from LTE eNBs is independent of the Wi-Fi network. 
	
	Since the interfering AP process is a non-independent thinning of $\Phi_W$, the Laplace transform of Wi-Fi interference (i.e., the second term in step (c)) is not known in closed-from. Therefore, similar to~\cite{nguyen2007dense80211,baccelli2010stochasticpt2}, we approximate the Wi-Fi interferers as a non-homogeneous PPP with intensity $\lambda_W h_1(r_0,x)$, which gives (\ref{COPCondiW1Wi-FiEq}).
\end{proof}
\begin{figure}
	\centering
	\includegraphics[height=2in, width=3.2in]{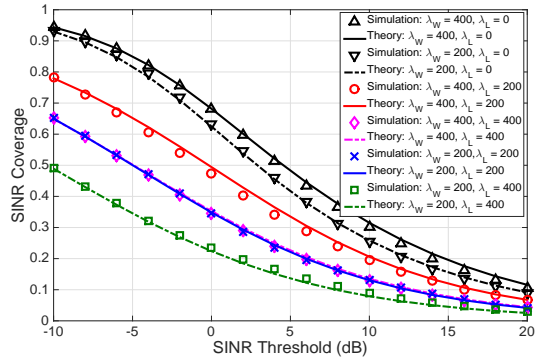}
	\caption{SINR coverage for the typical Wi-Fi STA.}\label{GlobCOPW1Wi-FiFig}
\end{figure}

\begin{remark}
For the rest of the paper, given the tagged AP or tagged eNB located at $(r_0,0)$ transmits, we use ``non-homogeneous PPP approximation'' to refer to the process of approximating Wi-Fi/LTE interferers as a non-homogeneous PPP with intensity $\lambda_W h(r_0,x)$/$\lambda_L h(r_0,x)$, where $h$ denotes the conditional MAP of the AP/eNB located at $x$. This will be used to derive the SINR distribution under various coexistence scenarios.
\end{remark}





Based on the parameters in Table~\ref{SysParaTable}, Fig.~\ref{GlobCOPW1Wi-FiFig} gives the SINR coverage performance of the typical Wi-Fi STA, where the simulation results are obtained from the definition of SINR in~(\ref{SINRWiFiW1Eq}). It can be observed from Fig.~\ref{GlobCOPW1Wi-FiFig} that Lemma~\ref{COPW1Wi-FiLemma} provides an accurate estimate of the actual SINR coverage. When $\lambda_L = 0$, Wi-Fi achieves good SINR performance due to the carrier sensing for Wi-Fi interferers. However, when coexisting with LTE, the additional interference contributed by the consistently transmitting LTE eNBs significantly impacts the SINR coverage of the typical Wi-Fi STA. 
The smaller the AP intensity $\lambda_W$, the more significant the LTE interference, which will lead to worse Wi-Fi SINR coverage performance. In Fig.~\ref{GlobCOPW1Wi-FiFig}, given $\lambda_L$, the Wi-Fi SINR coverage for $\lambda_W = 200$ is worse than the case when $\lambda_W = 400$. 


\subsubsection{SINR Coverage Probability of Typical LTE UE} The SINR coverage probability of the typical UE, which is denoted by $p_0^L(T, \lambda_W,\lambda_L)$, is given in the following lemma:

\begin{lemma}\label{SINRLTEW1Lemma}
	The SINR coverage probability for a typical LTE UE with SINR threshold $T$ can be approximated by:
	\allowdisplaybreaks
	\begin{align}\label{SINRLTEW1ApproxEq}
	\allowdisplaybreaks
	&p_{0}^L(T,\lambda_W,\lambda_L) \nonumber\\
	\approx & \int_{0}^{\infty}\exp\left(-\mu T \mathit{l}(r_0) \frac{\sigma_{N}^2}{P_L}\right) \exp\biggl( - \int_{\mathbb{R}^2 \setminus B(0,r_0)}  \frac{T\lambda_L\mathit{l}(r_0) {\rm d}y}{T \mathit{l}(r_0) + \mathit{l}(\|y\|)} \biggr) \exp\biggl(-\int_{\mathbb{R}^2} \frac{T\mathit{l}(r_0)\lambda_Wh_1^W(r_0,x)}{T\mathit{l}(r_0) + \frac{P_L}{P_W} \mathit{l}(\|x\|)} \biggr) \nonumber \\
	& \times f_{L}(r_0) {\rm d}r_0
	\end{align}where $h_1^W(r_0,x) = \frac{1-\exp(-N^W)}{N^W} \exp(-N^L_0(x ,r_0, \Gamma_{ed})) (1-\exp(-\mu \frac{\Gamma_{ed}}{P_L} \mathit{l}(\|y_0-x\|))) $ denotes the conditional MAP for AP $x$ given the tagged eNB $y_0 = (r_0,0)$ transmits. 
\end{lemma}

\begin{proof} According to Remark~\ref{RemarkAngleInvariance}, given the tagged eNB is located at $y_0 = (r_0,0)$, denoting the conditional SINR coverage probability by $p_{0}^L(r_0,T,\lambda_W,\lambda_L)$, we have:
	\allowdisplaybreaks
	\begin{align*}
	\allowdisplaybreaks
	&p_{0}^L(r_0,T,\lambda_W,\lambda_L) \nonumber \\
	\overset{(a)}{=} &\mathbb{E}\biggl[ \exp(-\mu T\mathit{l}(r_0))(\frac{\sigma_{N}^2}{P_L}  +  \sum\limits_{y_{m} \in \Phi_L \setminus \{ y_{0}\}} \frac{ F_{m,0}^{L}}{ \mathit{l}(\|y_{m}\|)}+ \sum\limits_{x_{j} \in \Phi_{W}} \frac{P_W}{P_L}  \frac{F_{j,0}^{WL} e_{j}^W}{ \mathit{l}(\|x_{j}\|)}) \biggl| y_0 \in \Phi_L , \Phi_L(B^0(0,r_0)) = 0 \biggr]  \nonumber \\
	\overset{(b)}{=} &\mathbb{E} \biggl[ \exp(-\mu T\mathit{l}(r_0))(\frac{\sigma_{N}^2}{P_L}  +  \sum\limits_{y_{m}^{L} \in \Phi_L \cap B^c(0,r_0) } \frac{F_{m,0}^{L}}{ \mathit{l}(\|y_{m}\|)}+ \sum\limits_{x_{j} \in \Phi_{W}  } \frac{P_W}{P_L} \frac{ F_{j,0}^{WL} \hat{e}_{j}^W} {\mathit{l}(\|x_{j}\|)}) \biggr] \nonumber\\ 
	\overset{(c)}{\approx}   &\exp(-\mu T\mathit{l}(r_0) \frac{\sigma_{N}^2}{P_L}) \mathbb{E} \biggl[ \exp\biggl(-\mu T\mathit{l}(r_0) \!\!\!\!\!\!\sum\limits_{y_{m} \in \Phi_L \cap B^c(0,r_0) }\!\!\! \frac{F_{m,0}^{L}}{\mathit{l}(\|y_{m}\|)}\biggr) \biggr] \mathbb{E} \biggl[ \exp\biggl(-\mu T\mathit{l}(r_0) \!\!\sum\limits_{x_{j} \in \Phi_{W} } \!\!\frac{P_W}{P_L}  \frac{F_{j,0}^{WL} \hat{e}_{j}^W}{ \mathit{l}(\|x_{j}\|)}\biggr) \biggr],
	\end{align*}where (a) is because the channels have Rayleigh fading and $y_0^L$ is the closest eNB to the typical user. Step (b) is obtained by using Slyvniak's theorem and de-conditioning on $\Phi_L(B^0(0,r_0)) = 0$. The modified medium access indicator for each AP in step (b) is given by:
	\begin{align*}
	\hat{e}_j^W = & \!\!\!\!\!\!\!\prod\limits_{y_{k} \in \Phi_{L} \cap B^c(0,r_0)} \!\!\!\biggl(\mathbbm{1}_{G_{kj}^{LW}/ \mathit{l}(\|y_{k} - x_{j}\|) \leq \frac{\Gamma_{ed}}{P_L }}\mathbbm{1}_{G_{0j}^{LW}/ \mathit{l}(\|y_{0} - x_{j}\|) \leq \frac{\Gamma_{ed}}{P_L }}\biggl) \!\!\!\!\! \prod\limits_{x_{i} \in \Phi_{W} \setminus \{x_{j}\}} \!\!\! \biggl(\mathbbm{1} _{t_{i}^{W} \geq t_{j}^{W}} \nonumber +\mathbbm{1} _{t_{i}^{W} < t_{j}^{W}} \mathbbm{1}_{ G_{ij}^{W}/ \mathit{l}(\|x_{i}-x_{j}\|) \leq \frac{\Gamma_{cs}}{P_W}} \biggl).
	\end{align*}
	For each Wi-Fi AP $x_j \in \Phi_W$, its modified MAP given the tagged eNB is at $y_0 = (r_0,0)$ is:
	\begin{align}\label{SINRLTEW1ProofEq2}
	&\mathbb{P}_{\Phi_W}^{x_j} (\hat{e}_{j}^W = 1) = \frac{1-\exp(-N^{W})}{N^{W}} \exp(-N^L_0(x_j,r_0,\Gamma_{ed}))  (1-\exp(-\mu \frac{\Gamma_{ed}}{P_L} \mathit{l}(\|y_0-x_j\|)),
	\end{align}where $N^W$ and $N^L_0(x_j,r_0,\Gamma_{ed})$ are defined in Table~\ref{FunctionTable}. In step (c), the correlation between the interference from eNBs and APs is neglected for simplicity. Finally, the desired result is obtained by treating $\hat{e}_j^W$ as independent for each AP $x_j$, and applying the non-homogeneous PPP approximation to  Wi-Fi interferers. 	
\end{proof}
\begin{figure}
	\centering
	\includegraphics[height=2in, width=3.2in]{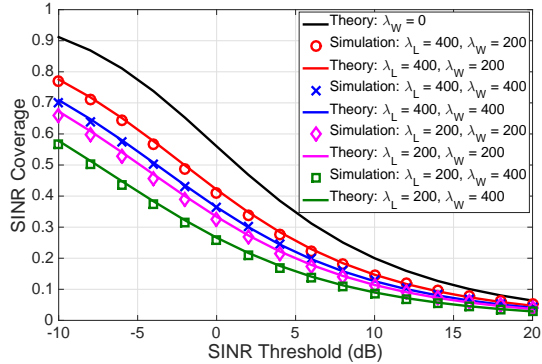}
	\caption{SINR coverage for the typical LTE UE.}\label{GlobCOPW1LTEFig}
\end{figure}
\begin{remark}
	The first two terms in~(\ref{SINRLTEW1ApproxEq}) come from the thermal noise and the eNB interferers respectively, which give the same result as Theorem 2 in~\cite{trac}. In contrast, the effect of coexisting transmitting Wi-Fi APs is reflected in the third term, 
	which decreases by increasing the Wi-Fi AP intensity $\lambda_W$ or the energy detection threshold $\Gamma_{ed}$.
\end{remark}

The SINR coverage for the typical LTE UE is evaluated 
in Fig.~\ref{GlobCOPW1LTEFig}, where the simulation results are obtained from the following procedure. First, 50 realizations for eNB PPP and 50 realizations for AP PPP are generated in an 1 km $\times$ 1 km area, which results in a total of 2500 combinations of the eNB and AP processes. For each combination, we determine the medium access indicator at each AP according to~(\ref{MAIWi-FiW1}). In addition, 50 uniformly chosen locations for the typical STA are generated, and we evaluate the received SINR at each STA location if the serving AP transmits. Finally, the SINR coverage probability is obtained as the fraction of STAs whose received SINR exceeds the threshold $T$.

The SINR coverage performance when $\lambda_W = 0$ is provided in Fig.~\ref{GlobCOPW1LTEFig}. It is independent of the eNB intensity under Rayleigh fading channels for negligible thermal noise power~\cite{trac}. From Fig.~\ref{GlobCOPW1LTEFig}, the accuracy of Lemma~\ref{SINRLTEW1Lemma} can be validated. In addition, it can be observed that the typical LTE UE achieves better SINR coverage when increasing the eNB intensity $\lambda_L$ or decreasing the AP intensity $\lambda_W$, with the SINR coverage for $\lambda_W = 0$ as an upper bound. In particular, when $\lambda_L = \lambda_W$, the MAP for each Wi-Fi AP becomes smaller by increasing $\lambda_L$, and therefore a better SINR coverage can be achieved by the LTE UE when $\lambda_L$ is larger. Overall, it can be observed from Fig.~\ref{GlobCOPW1LTEFig} that unless $\lambda_L \ll \lambda_W$, the typical LTE UE achieves reasonable SINR performance compared to the case when $\lambda_W = 0$, which demonstrates the robustness of the LTE system to the coexisting Wi-Fi system. 

Therefore, when LTE coexists with Wi-Fi without any protocol changes, LTE is able to maintain good SINR coverage performance, while Wi-Fi experiences drastically degraded SINR coverage. This imbalanced performance means some fair coexistence methods have to be implemented by LTE in order to guarantee a reasonable performance for Wi-Fi network. The DST and rate coverage performance for LTE and  Wi-Fi can be derived directly from Corollary~\ref{MAPWi-FitagW1lemma}, Lemma~\ref{COPW1Wi-FiLemma} and Lemma~\ref{SINRLTEW1Lemma}. The detailed discussions are provided in Section~\ref{NumEvaSec}.

Finally, although we consider a downlink only scenario for Wi-Fi, similar techniques can
be used to derive the MAP and SINR coverage performance when Wi-Fi uplink traffic also
exists. Since STAs will apply the same channel access mechanism as APs, the medium access indicator for each AP and STA will account for both the contending APs and STAs. The detailed performance when Wi-Fi uplink traffic exists is left to future work.

\section{LTE with Discontinuous Transmission}\label{LTEwDT}
A straightforward scheme to guarantee the fair-coexistence between Wi-Fi and LTE is to let LTE adopt a discontinuous, duty-cycle transmission pattern, which is also know as LTE-U~\cite{qrc2014LTE,lteu2015TRv1}. Specifically, LTE transmits for a fraction $\eta$ of time ($0 \leq \eta \leq 1 $), and is muted for the complementary 1-$\eta$ fraction. 

The LTE transmission duy cycle $\eta$ can be fixed or adaptively adjusted based on Wi-Fi medium utilization~\cite{qrc2014LTE}. Generally, $\eta$ needs to be chosen in such a way that LTE shall not impact Wi-Fi more than an additional Wi-Fi network w.r.t. SINR coverage probability, rate coverage, etc. 
We consider a static muting pattern for LTE, where all the eNBs follow the same muting pattern either synchronously or asynchronously. If the eNBs are muted synchronously, they transmit and mute at the same time. If the eNBs are muted asynchronously, the neighboring eNBs could adopt a shifted version of the muting pattern~\cite{jeon2014lteworkshop}. For simplicity, we assume each eNB is transmitting with probability $\eta$ at a given time under the asynchronous scheme. In the rest of this section, the time-averaged DST and rate coverage performance when LTE transmits discontinuously are derived. 
\subsection{LTE with Synchronous Discontinuous Transmission Pattern}
In this case, since all eNBs transmit and mute at the same time, the MAP for the tagged Wi-Fi AP during LTE ``On" and ``Off" period are $\hat{p}_{0,MAP}^W(\lambda_W,\lambda_L)$ and $\hat{p}_{0,MAP}^W(\lambda_W,0)$ respectively, where $\hat{p}_{0,\text{MAP}}^W$ is given in~(\ref{MAPWi-FitageW1Eq}). 
Similarly, the SINR coverage probability of the typical Wi-Fi STA (resp. LTE UE) with threshold $T$ is $p_0^W(T,\lambda_W,\lambda_L)$ (resp. $p_0^L(T,\lambda_W,\lambda_L)$) and $p_0^W(T,\lambda_W,0)$ (resp. 0) during LTE ``On" and ``Off" period respectively, where $p_0^W$ and $p_0^L$ are provided in Lemma~\ref{COPW1Wi-FiLemma} and Lemma~\ref{SINRLTEW1Lemma}. Define the time-averaged DST with SINR threshold $T$ as the time-averaged fraction of links that can support SINR level $T$.
\begin{lemma}\label{DSTW2SyncLemma}
	When LTE adopts a synchronous discontinuous transmission pattern with duty cycle $\eta$, the time-averaged DST with threshold $T$ for the Wi-Fi and LTE network are given by:
	\allowdisplaybreaks
	\begin{align}\label{DSTW2Wi-FiEq}
	\allowdisplaybreaks
		&d_{1,suc}^{W}(\lambda_W,\lambda_L,T,\eta) = \eta \lambda_W \hat{p}_{0,MAP}^W(\lambda_W,\lambda_L)  p_0^W(T,\lambda_W,\lambda_L) + (1-\eta) \lambda_W \hat{p}_{0,MAP}^W(\lambda_W,0)  p_0^W(T,\lambda_W,0),\nonumber\\
		&d_{1,suc}^{L}(\lambda_W,\lambda_L,T,\eta) = \eta \lambda_L  p_0^L(T,\lambda_W,\lambda_L).
	\end{align}
\end{lemma}
\begin{proof}
	Since LTE transmits for $\eta$ fraction of time and silences for $1-\eta$ fraction time, the time-averaged DST performance for Wi-Fi and LTE can be obtained directly from Definition~\ref{DSTDefn}.
\end{proof}
In addition, the time-averaged rate coverage probability with threshold $\rho$ is defined as the time-averaged fraction of eNBs/APs that can support an aggregate data rate of $\rho$. Since each LTE eNB transmits for $\eta$ fraction of time, we treat the MAP of the tagged eNB as $\eta$ in~(\ref{RCOPDefnEq}). 
\begin{lemma}\label{RCOPW2SyncLemma}
	When LTE adopts a synchronous discontinuous transmission pattern with duty cycle $\eta$, the time-averaged rate coverage probability with rate threshold $\rho$ for Wi-Fi and LTE are given by:
	\begin{align}\label{RCOPW2Wi-FiEq}
	&P_{1,rate}^{W}(\lambda_W,\lambda_L,\rho,\eta) = \eta p_0^W(2^{\frac{\rho}{B\hat{p}_{0,MAP}^W(\lambda_W,\lambda_L) }}-1,\lambda_W,\lambda_L) + (1-\eta) p_0^W(2^{\frac{\rho}{B\hat{p}_{0,MAP}^W(\lambda_W,0) }}-1,\lambda_W,0),\nonumber\\
	&P_{1,rate}^{L}(\lambda_W,\lambda_L,\rho,\eta) =  p_0^L(2^{\frac{\rho}{B \eta }}-1,\lambda_W,\lambda_L).
	\end{align}
\end{lemma}

\begin{proof}
	The time-averaged Wi-Fi rate coverage can be derived since the fraction of Wi-Fi APs that can support data rate $\rho$ is $p_0^W(2^{\frac{\rho}{B\hat{p}_{0,MAP}^W(\lambda_W,\lambda_L)}}-1,\lambda_W,\lambda_L)$ and $p_0^W(2^{\frac{\rho}{B\hat{p}_{0,MAP}^W(\lambda_W,0)}}-1,\lambda_W,0)$ during LTE ``on" and ``off" period respectively. In addition, the time-averaged LTE rate coverage is derived by noting that the typical LTE link is active for $\eta$ fraction of time.
\end{proof}

	It is straightforward from~(\ref{DSTW2Wi-FiEq}) and~(\ref{RCOPW2Wi-FiEq}) that better DST and rate coverage can be achieved by Wi-Fi when $\eta$ decreases. By contrast, since $p_0^L(T,\lambda_W,\lambda_L)$ is a decreasing function w.r.t. the SINR threshold $T$, LTE achieves better DST and rate coverage when $\eta$ increases.
\subsection{LTE with Asynchronous Discontinuous Transmission Pattern}
Since each eNB transmits independently with probability $\eta$ at a given time, the eNBs contributing to the interference of Wi-Fi form a PPP with intensity $\eta \lambda_L$. Therefore, the MAP for the tagged AP is $\hat{p}_{0,\text{MAP}}^W(\lambda_W, \eta \lambda_L)$, and the SINR coverage probability with threshold $T$ for the typical Wi-Fi STA is $p_{0}^W(T,\lambda_W, \eta \lambda_L)$. Correspondingly, the time-averaged DST of Wi-Fi is given by:
\allowdisplaybreaks
\begin{align}\label{DSTW2Wi-FiAsyncEq}
\allowdisplaybreaks
d_{2,suc}^{W}(\lambda_W,\lambda_L,T,\eta) = \lambda_W \hat{p}_{0,MAP}^W(\lambda_W,\eta \lambda_L)  p_0^W(T,\lambda_W,\eta \lambda_L),
\end{align}
and the time-averaged rate coverage probability of Wi-Fi is given by:
\allowdisplaybreaks
\begin{align}\label{RCOPW2Wi-FiAsyncEq}
\allowdisplaybreaks
P_{2,rate}^{W}(\lambda_W,\lambda_L,\rho,\eta) =  p_0^W(2^{\frac{\rho}{B\hat{p}_{0,MAP}^W(\lambda_W, \eta \lambda_L) }}-1,\lambda_W,\eta \lambda_L).
\end{align}
According to~(\ref{DSTW2Wi-FiAsyncEq}) and~(\ref{RCOPW2Wi-FiAsyncEq}), Wi-Fi achieves better DST and rate coverage when $\eta$ decreases.

For LTE, during the $\eta$ fraction of time that the tagged eNB transmits, the interfering eNBs form a PPP with intensity $\eta \lambda_L$. Thus, the time-averaged DST of LTE is given by:
\begin{align}\label{DSTW2LTEAsyncEq}
d_{2,suc}^{L}(\lambda_W,\lambda_L,T,\eta) = \lambda_L \eta \int_{0}^{\infty} p_0^L(r_0, T,\lambda_W,\eta \lambda_L) 2\pi \lambda_L r_0 \exp(-\lambda_L \pi r_0^2) {\rm d}r_0,
\end{align}
and the time-averaged rate coverage probability is given by:
\begin{align}\label{RCOPW2LTEAsyncEq}
P_{2,rate}^{L}(\lambda_W,\lambda_L,\rho,\eta) =   \int_{0}^{\infty} p_0^L(r_0, 2^{\frac{\rho}{B \eta}}-1,\lambda_W,\eta \lambda_L) 2\pi \lambda_L r_0 \exp(-\lambda_L \pi r_0^2) {\rm d}r_0,
\end{align}
where $p_0^L(r_0, T,\lambda_W,\lambda_L)$ is derived in Lemma~\ref{SINRLTEW1Lemma}.
\begin{figure}
	\begin{subfigure}[b]{0.5\textwidth}
		\includegraphics[height=2in, width=3.1in]{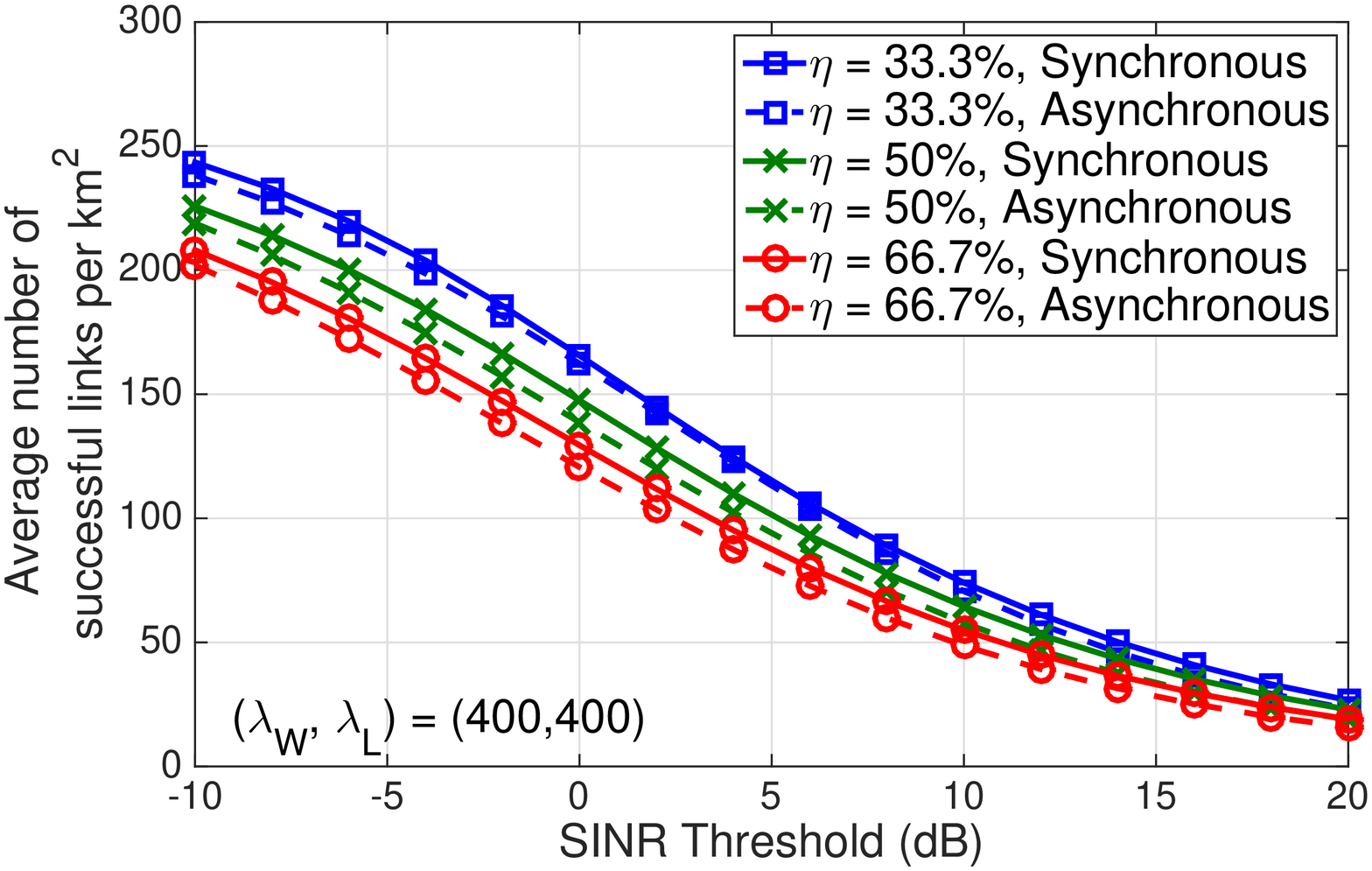}
		\caption{DST for Wi-Fi}\label{Wi-FiDTFig}
	\end{subfigure}
	\hfill
	\begin{subfigure}[b]{0.5\textwidth}
		\includegraphics[height=2in, width=3.1in]{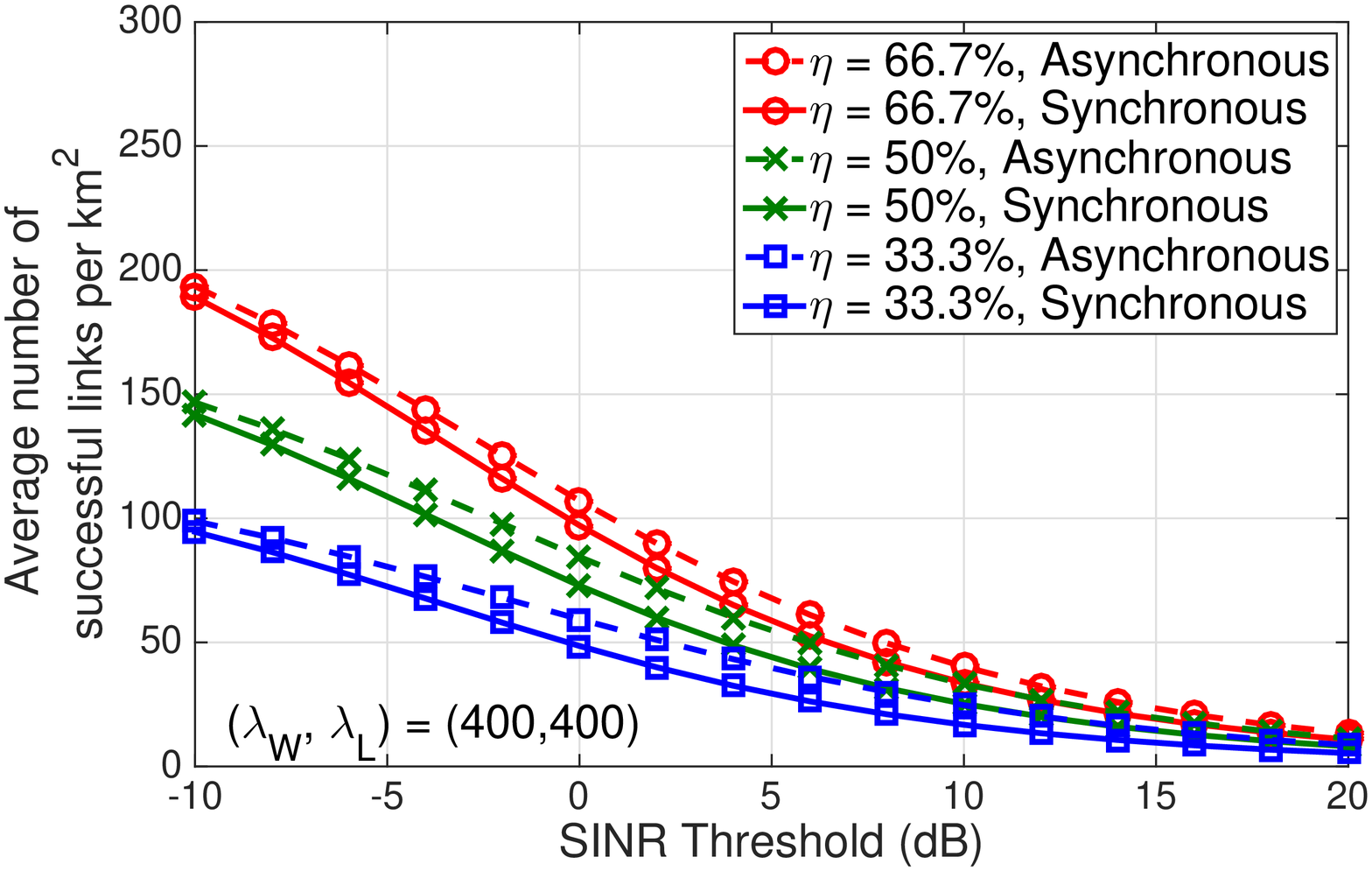}
		\caption{DST for LTE}\label{LTEDTFig}
	\end{subfigure}
	\caption{DST comparison of the synchronous and asynchronous muting pattern.}\label{LTEwDTFig}
\end{figure}

\begin{figure}
	\begin{subfigure}[b]{0.5\textwidth}
		\includegraphics[height=2in, width=3.1in]{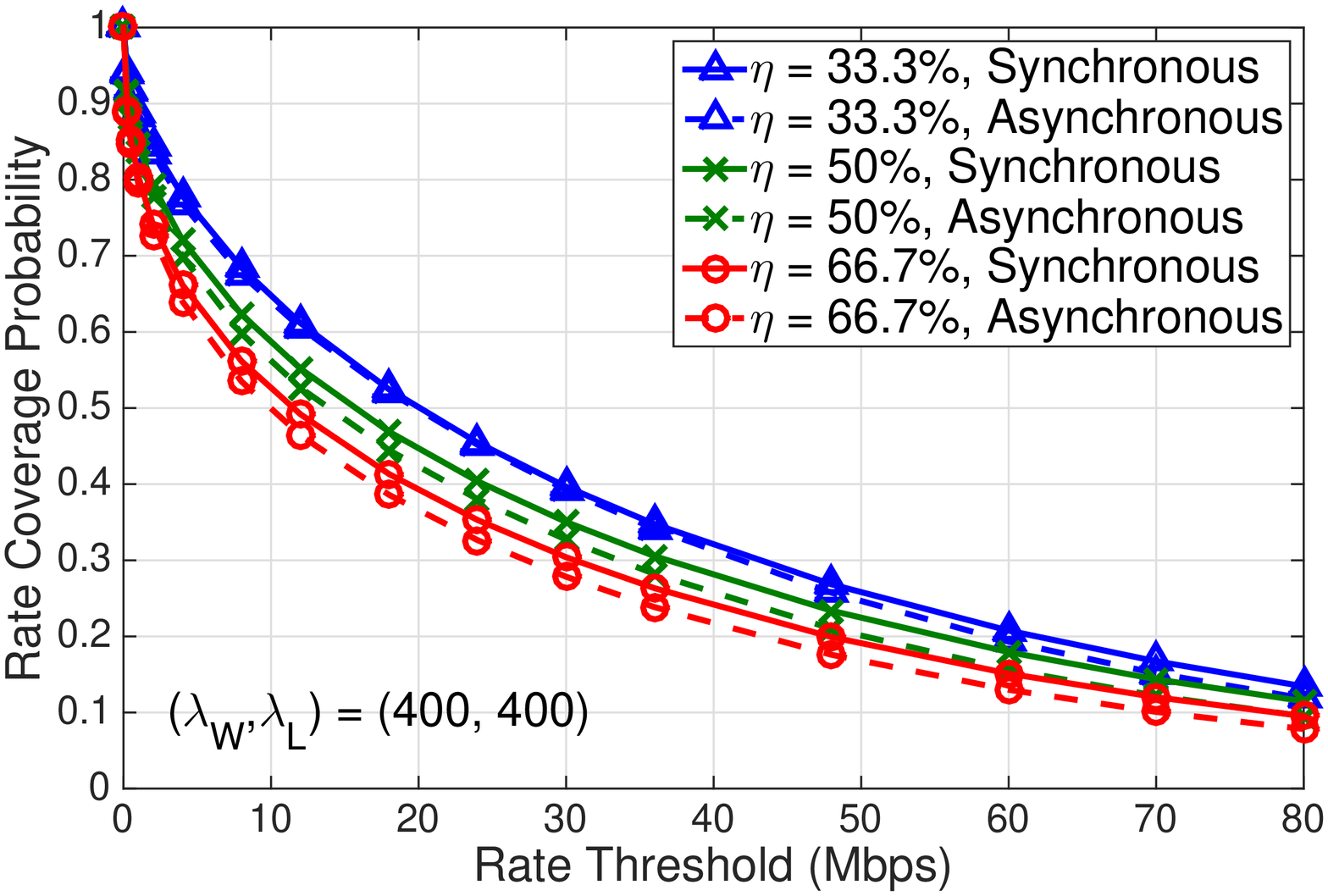}
		\caption{Rate coverage for Wi-Fi}\label{Wi-FiRateCompDTFig}
	\end{subfigure}
	\hfill
	\begin{subfigure}[b]{0.5\textwidth}
		\includegraphics[height=2in, width=3.1in]{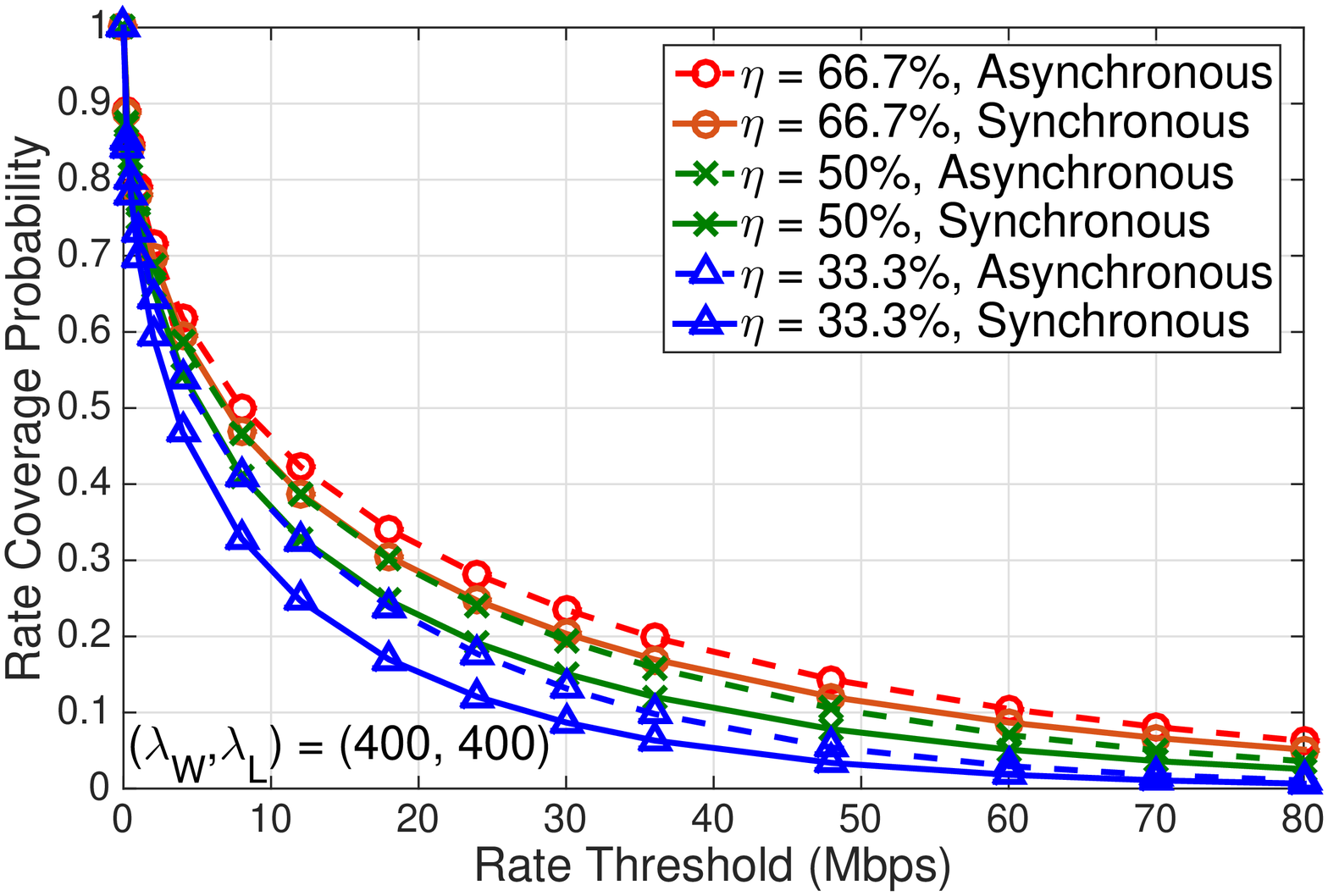}
		\caption{Rate coverage for LTE}\label{LTERateCompDTFig}
	\end{subfigure}
	\caption{Rate coverage comparison of the synchronous and asynchronous muting pattern.}\label{RateCompDTFig}
\end{figure}

\subsection{Comparison of Synchronous and Asynchronous Muting Patterns}
Fig.~\ref{LTEwDTFig} and Fig.~\ref{RateCompDTFig} show the analytical time-averaged DST and rate coverage performance when $\lambda_W = 400$ APs/km$^2$ and $\lambda_L = 400$ eNBs/km$^2$. In terms of Wi-Fi DST and rate coverage performance, the synchronous LTE muting pattern generally outperforms the asynchronous one. This is due to fact that when all LTE eNBs are muted, Wi-Fi APs observe a much cleaner channel and therefore benefit more from LTE muting compared to the asynchronous scheme. Since LTE interferers form an independent thinning of the eNB process under the asynchronous muting pattern, the latter outperforms the synchronous pattern in terms of DST and rate coverage.
In addition, Fig.~\ref{LTEwDTFig} and Fig.~\ref{RateCompDTFig} also indicate that LTE needs to adopt a short transmission duty cycle $\eta$ (e.g., less than 50\%) to protect Wi-Fi. However, LTE is also more sensitive to the transmission duty cycle compared to Wi-Fi, which means that a very small $\eta$ leads to much degraded performance of LTE. Therefore, a synchronous muting pattern with a reasonably short LTE transmission duty cycle (e.g., within 33.3\% to 50\%) is suggested to protect Wi-Fi.

\section{LTE with Listen-before-talk and Random Backoff}\label{LTEwLBTBFSec}
Besides LTE with discontinuous transmission, another fair coexistence method is to let LTE implement the listen-before-talk (LBT) and random backoff (BO) mechanism similar to Wi-Fi. 
Specifically, we consider each eNB implements carrier sense mechanism to detect strong interfering LTE and Wi-Fi neighbors with a common threshold $\Gamma^L$. In addition, each eNB implements a random back off timer, which is uniformly distributed between $a$ and $b$. The value of $(a,b)$ determines how aggressively LTE eNBs access the channel. The medium access indicators for AP $x_i$ and eNB $y_k$ are given as follows:
\allowdisplaybreaks
\begin{align}\label{MAIEqW4}
\allowdisplaybreaks
e_{i}^{W} = & \!\!\!\!\!\!\prod\limits_{x_{j} \in \Phi_{W} \setminus \{x_{i}\}} \!\!\!\biggl(\mathbbm{1} _{t_{j}^{W} \geq t_{i}^{W}} +\mathbbm{1} _{t_{j}^{W} < t_{i}^{W}} \mathbbm{1}_{ G_{ji}^{W}/ \mathit{l}(\|x_{j} - x_{i}\|) \leq \frac{\Gamma_{cs}}{P_W}} \biggl) \!\!\prod\limits_{y_{m} \in \Phi_{L}} \biggl(\mathbbm{1} _{t_{m}^{L} \geq t_{i}^{W}} +\mathbbm{1} _{t_{m}^{L} < t_{i}^{W}} \mathbbm{1}_{G_{mi}^{LW}/ \mathit{l}(\|y_{m} - x_{i}\|) \leq \frac{\Gamma_{ed}}{P_L}} \biggl), \nonumber\\
e_{k}^{L} = & \!\!\prod\limits_{x_{j} \in \Phi_{W}} \!\!\biggl(\mathbbm{1} _{t_{j}^{W} \geq t_{k}^{L} } +\mathbbm{1} _{ t_{j}^{W} < t_{k}^{L} } \mathbbm{1}_{ G_{jk}^{WL}/ \mathit{l}(\|y_{k}-x_{j}\|) \leq \frac{\Gamma^{L}}{P_W}} \biggl)  \!\!\!\!\prod\limits_{y_{m} \in \Phi_{L} \setminus \{y_{k}\}} \!\!\biggl(\mathbbm{1}_{t_{m}^{L} \geq t_{k}^{L}} +\mathbbm{1} _{t_{m}^{L} < t_{k}^{L}} \mathbbm{1}_{ G_{mk}^{L}/ \mathit{l}(\|y_{m} - y_{k}\|) \leq \frac{\Gamma^{L}}{P_L}} \biggl).
\end{align}The expression for $e_i^W$ means AP $x_i$ dose not transmit whenever the power it receives from any AP or eNB with a smaller back-off timer exceeds $\Gamma_{cs}$ or $\Gamma_{ed}$; while the expression for $e_k^L$ means eNB $y_k$ does not transmit whenever the power it receives from any AP or eNB with a smaller back-off timer exceeds $\Gamma^L$. By implementing the LBT and random BO scheme, LTE has some flexibility in choosing the sensing threshold (i.e., $\Gamma^{L}$), and the channel access priority (i.e., $(a,b)$) to better coexist with Wi-Fi. In particular, two channel access priority scenarios of LTE are considered, namely when LTE has the same channel access priority as Wi-Fi, and when LTE has the lower channel access priority than Wi-Fi. These two scenarios correspond to the cases when $(a,b) = (0,1)$ and $(a,b) = (1,2)$, which are analyzed in the rest of the section. 

\subsection{LTE with Same Channel Access Priority as Wi-Fi when $(a,b) = (0,1)$}
In this case, 
LTE has the same channel access priority as Wi-Fi in terms of the random backoff procedure. In addition, the sensitivity of LTE to interfering signals is controlled by the threshold $\Gamma^{L}$. A more sensitive $\Gamma^{L}$ provides a better protection to Wi-Fi, and vice versa. The MAP for a typical AP and eNB can be easily derived from~(\ref{MAIEqW4}):
\begin{lemma}\label{MAPW4Case1Lemma}
	When LTE follows the LBT and random BO mechanism with $(a,b) = (0,1)$, the MAPs for typical AP and eNB, denoted by $p_{3,MAP}^{W}$ and $p_{3,MAP}^{L}$ respectively, are given by:
	\begin{align*} 
	p^{W}_{3,\text{MAP}} (\lambda_W,\lambda_L)&= \frac{1-\exp(-N^W-N^L)}{N^W+N^L},\\
	p^{L}_{3,\text{MAP}} (\lambda_W,\lambda_L) &= \frac{1-\exp(-N^W_3(\Gamma^{L})-N^L_3(\Gamma^{L}))}{N^W_3(\Gamma^{L}) + N^L_3(\Gamma^{L})},
	\end{align*}where $N^W$, $N^L$, $N^W_3(\Gamma^{L})$ and $N^L_3(\Gamma^{L})$ are defined in Table~\ref{FunctionTable}.
\end{lemma}

Lemma~\ref{MAPW4Case1Lemma} can be proved similarly to Lemma~\ref{MAPWi-FiW1lemma}; thus the detailed proof is omitted.

\begin{remark}
It is straightforward from Lemma~\ref{MAPW4Case1Lemma} that $\frac{1}{1 + N^W + N^L} \leq p^{W}_{3,\text{MAP}}(\lambda_W,\lambda_L) < \frac{1}{N^W + N^L}$. Therefore, the MAP for the typical Wi-Fi AP is inversely proportional to the total number of its Wi-Fi and LTE contenders. Similarly, the MAP for the typical eNB is inversely proportional to the total number of APs and eNBs whose power received by the typical eNB exceeds $\Gamma^L$.
\end{remark}

\begin{corollary}\label{MAPTagAPW4Case1Lemma}
	When LTE implements LBT and random BO with contention window size $[a,b] = [0,1]$, the MAPs of the tagged Wi-Fi AP and LTE eNB are given by:
		\allowdisplaybreaks
		\begin{align*} 
		\allowdisplaybreaks
		\mathbb{E}(e_0^W)&= \int_{0}^{\infty}\frac{1-\exp(-N^W_2(r_0)-N^L)}{N^W_2(r_0)+N^L} f_{W}(r_0)  {\rm d}r_0,\\
		\mathbb{E}(e_0^L) &= \int_{0}^{\infty} \frac{1-\exp(-N^W_3(\Gamma^L) - N^L_1(r_0,\Gamma^L))}{N^W_3(\Gamma^L) + N^L_1(r_0,\Gamma^L)} f_{L}(r_0) {\rm d}r_0.
		\end{align*}
\end{corollary}
In terms of MAP, the effect of the additional LTE network on Wi-Fi is similar to that of deploying another CSMA/CA network with intensity $\lambda_L$ and transmit power $P_L$. However, since each STA (UE) can only associate with its closest AP (eNB), the LTE (Wi-Fi) network becomes a closed access CSMA/CA network to Wi-Fi (LTE), which may have significant impact on SINR performance. 
Since the transmitting eNB/AP process is a dependent thinning of $\Phi_L$/$\Phi_W$, whose Laplace functional is generally unknown in a closed form, the independent non-homogeneous PPP approximation to the transmitting eNB and AP point processes is used. First, we derive the following conditional MAP for each AP/eNB given the tagged AP transmits. 
\begin{corollary}\label{CondiRetainProbW4Wi-FiLemma}
	Conditionally on the fact that tagged AP $x_0 =(r_0,0)$ transmits, the probability for another AP $x_i \in \Phi_W \cap B^c(0,r_0)$ to transmit is:
	\allowdisplaybreaks
	\begin{align} \label{CondiRetainProbW4Wi-FiEq}
	\allowdisplaybreaks
	 h_2^W(r_0,x_i) = \frac{V(x_i-x_0,\frac{\Gamma_{cs}}{P_W}, \frac{\Gamma_{cs}}{P_W},N_1,N_2,N_3)}{U(x_i-x_0,\frac{\Gamma_{cs}}{P_W},N_2)},
	\end{align}where $N_1= N^W_0(x_i,r_0,\Gamma_{cs})+ N^L$, $N_2= N^W_2(r_0)+ N^L$ and $N_3 =  C^W_1(x_0,x_i)+ C^L_2(x_i-x_0)$.
\end{corollary}

\begin{corollary}\label{CondiRetainProbW4LTELemma}
	Conditionally on the fact that tagged AP $x_0 = (r_0,0)$ transmits, the probability for eNB $y_k \in \Phi_L$ to transmit is:
		\allowdisplaybreaks
		\begin{align}\label{CondiRetainProbW4LTEEq}
		h_2^{L}(r_0,y_k) = \frac{V(y_k-x_0, \frac{\Gamma^L}{P_W}, \frac{\Gamma_{ed}}{P_L},N_4,N_5,N_6)}{U(y_k-x_0,\frac{\Gamma_{ed}}{P_L},N_5)},
		\end{align}where $N_4 = N^W_0(y_k,r_0,\Gamma^L) + N^L_3(\Gamma^L)$, $N_5 = N^W_2(r_0) + N^L$, and $N_6 = C^W_0(y_k,\Gamma^L,x_0,\Gamma_{cs}) + C^L_0(y_k-x_0,\Gamma^L,o,\Gamma_{ed})$.
\end{corollary}

The proof of Corollary~\ref{CondiRetainProbW4Wi-FiLemma} is provided in the Appendix~\ref{CondiMAPW4Appdx}, while Corollary~\ref{CondiRetainProbW4LTELemma} can be proved in a similar way to Corollary~\ref{CondiRetainProbW4Wi-FiLemma}; thus we omit the detailed proof. 

Given the tagged AP $x_0$ transmits, we resort to approximating the interfering AP and eNB process by two independent non-homogeneous PPPs with intensities $\lambda_W h_2^W(r_0,x)$ and $\lambda_L h_2^L(r_0,x)$ respectively, which leads to the following approximate SINR coverage of the typical Wi-Fi STA:
\begin{lemma}\label{CondiCOPW4Wi-FiLemma}
	When LTE implements the listen-before-talk and random backoff mechanism with $(a,b) = (0,1)$, the approximate SINR coverage probability of the typical STA is given by:
	\allowdisplaybreaks
	\begin{align}\label{CondiCOPW4Wi-FiEq}
	\allowdisplaybreaks
	p_3^W(T,\lambda_W,\lambda_L) \approx & \int_{0}^{\infty}\exp\biggl( -\mu T \mathit{l}(r_0) \frac{\sigma_{N}^2}{P_W}\biggr) \exp\biggl(  -\int_{\mathbb{R}^2} \frac{T \mathit{l}(r_0) \lambda_L h_{2}^{L}(r_0,y)}{\frac{P_W}{P_L}\mathit{l}(\|y\|) + T \mathit{l}(r_0)} {\rm d}y\biggr) \nonumber \\
	&\times \exp\biggl(- \int_{\mathbb{R}^2 \setminus B(0,r_0)} \frac{T\mathit{l}(r_0) \lambda_W h^{W}_{2}(r_0,x)}{\mathit{l}(\|x\|) + T \mathit{l}(r_0)}   {\rm d}x \biggl) f_{W}(r_0){\rm d}r_0.
	\end{align}
\end{lemma} 

Lemma~\ref{CondiCOPW4Wi-FiLemma} can be easily proved using the non-homogeneous PPP approximation; thus we omit the detailed proof. 

\begin{remark}
The first and second terms in~(\ref{CondiCOPW4Wi-FiEq}) are stems from thermal noise and LTE interferers respectively, while the third term stems from the transmitting Wi-Fi interferers. The intensity of the Wi-Fi interferers at $x \in \mathbb{R}^2 \cap B^{o}(0,r_0)$ is described by the function $\lambda_W h_2^W(r_0,x)$. Note that when $\|x\| \rightarrow \infty$, we have $N^W_0(x,r_0,\Gamma_{cs}) \rightarrow N^W$, $C^W_1(x_0,x) \rightarrow 0$ and $C^L_2(x-x_0) \rightarrow 0$, which gives the following asymptotic result: $	\lim\limits_{\|x\| \rightarrow \infty} \lambda_W h_2^W(r_0,x) = \lambda_W p^{W}_{3,\text{MAP}} (\lambda_W,\lambda_L).$
The intensity of LTE interferers also satisfies the asymptotic result: $\lim\limits_{\|y\| \rightarrow \infty} \lambda_L h_2^L(r_0,y) = \lambda_L p^{L}_{3,\text{MAP}} (\lambda_W,\lambda_L)$.
\end{remark}

Similar to Wi-Fi, given the tagged eNB is located at $y_0 = (r_0,0)$, the modified medium access indicators for each AP and eNB are given by:
\allowdisplaybreaks
\begin{align*}
\allowdisplaybreaks
\hat{e}_{i}^{W} = & \!\!\!\!\prod\limits_{x_{j} \in \Phi_{W} \setminus \{x_{i}\}} \!\!\!\!\!\!\biggl(\mathbbm{1} _{t_{j}^{W} \geq t_{i}^{W}} +\mathbbm{1} _{t_{j}^{W} < t_{i}^{W}} \mathbbm{1}_{G_{ji}^{W}/ \mathit{l}(\|x_{j} - x_{i}\|) \leq \frac{\Gamma_{cs}}{P_W }} \biggl)  \!\!\!\!\!\!\!\!\!\!\prod\limits_{y_{m} \in (\Phi_{L}  \cap B^c(0,r_0)+ \delta_{y_0})}\!\!\!\!\!\!\!\! \biggl(\mathbbm{1} _{t_{m}^{L} \geq t_{i}^{W}} +\mathbbm{1} _{t_{m}^{L} < t_{i}^{W}} \mathbbm{1}_{G_{mi}^{LW}/ \mathit{l}(\|y_{m} - x_{i}\|) \leq \frac{\Gamma_{ed}}{P_L }} \biggl),\nonumber\\
\hat{e}_{k}^{L} = & \!\!\prod\limits_{x_{j} \in \Phi_{W} }  \!\!\biggl(\mathbbm{1} _{t_{j}^{W} \geq t_{k}^{L} } +\mathbbm{1}_{ t_{j}^{W} < t_{k}^{L}} \mathbbm{1}_{G_{jk}^{WL}/ \mathit{l}(\|y_{k}-x_{j}\|) \leq \frac{\Gamma^L}{P_W }} \biggl)  \!\!\!\!\!\!\!\! \!\!\prod\limits_{y_{m} \in (\Phi_{L} \cap B^c(0,r_0)+ \delta_{y_0})  \setminus \{y_{k}\}} \!\!\!\!\!\!\!\!\!\! \biggl(\mathbbm{1}_{t_{m}^{L} \geq t_{k}^{L}} +\mathbbm{1}_{t_{m}^{L} < t_{k}^{L}} \mathbbm{1}_{G_{mk}^{L}/ \mathit{l}(\|y_{m} - y_{k}\|) \leq \frac{\Gamma^L}{P_L }} \biggl).
\end{align*}By following the same procedure as Corollary~\ref{CondiRetainProbW4Wi-FiLemma} and Corollary~\ref{CondiRetainProbW4LTELemma}, we can calculate the conditional MAP for each AP and eNB, given the tagged eNB of the typical UE transmits. 
\begin{corollary}\label{Wi-FiCondiRetainProbGivenLTEW4Coro}
	Conditionally on the fact that the tagged eNB $y_0 =(r_0,0)$ transmits, the probability for another AP $x_i \in \Phi_W \cap B^c(0,r_0)$ to transmit is:
	\allowdisplaybreaks
	\begin{align*}
	\allowdisplaybreaks
	h_3^{W}(r_0,x_i) = \frac{V(x_i-y_0,\frac{\Gamma_{ed}}{P_L},\frac{\Gamma^L}{P_W},N_1,N_2,N_3)}{U(x_i-y_0,\frac{\Gamma^L}{P_W},N_1)},
	\end{align*}where $N_1 = N^W_3(\Gamma^{L}) + N^L_1(r_0,\Gamma^{L})$, $N_2 = N^W + N^L_0(x_i,r_0,\Gamma_{ed})$, and $N_3 = C^W_0(y_0-x_i,\Gamma^{L},o,\Gamma_{cs}) + C^L_0(x_i,\Gamma_{ed},y_0,\Gamma^{L})$.
\end{corollary}

\begin{corollary}\label{LTECondiRetainProbGivenLTEW4Coro}
	Conditionally on the fact that the tagged eNB $y_0 =(r_0,0)$ transmits, the probability for another AP $x_i \in \Phi_W \cap B^c(0,r_0)$ to transmit is:
	\allowdisplaybreaks
	\begin{align*}
	\allowdisplaybreaks
	h_3^{L}(r_0,y_k) = \frac{V(y_k-y_0,\frac{\Gamma^L}{P_L},\frac{\Gamma^{L}}{P_L},N_4,N_5,N_6)}{U(y_k-y_0,\frac{\Gamma^L}{P_L},N_4)},
	\end{align*}where $N_4 = N^W_3(\Gamma^{L}) + N^L_1(r_0,\Gamma^{L})$, $N_5 = N^W_3(\Gamma^{L}) + N^L_0(y_k,r_0,\Gamma^{L})$, and $N_6 = C^W_0(y_0-y_k,\Gamma^{L},o,\Gamma^{L}) + C^L_0(y_k,\Gamma^{L},y_0,\Gamma^{L})$.
\end{corollary}

Based on Corollary~\ref{Wi-FiCondiRetainProbGivenLTEW4Coro} and Corollay~\ref{LTECondiRetainProbGivenLTEW4Coro}, the SINR coverage of the typical UE can also be derived using the non-homogeneous PPP approximation of the interfering eNBs and APs:
\begin{lemma}\label{COPLTEW4Case1Lemma}
	When LTE implements listen-before-talk and random backoff mechanism with $(a,b) = (0,1)$, the approximate SINR coverage probability of the typical LTE UE is:
	\begin{align*}
	\allowdisplaybreaks
	p_3^L(T,\lambda_W,\lambda_L)  \approx &\int_{0}^{\infty}\exp\biggl( -\mu T \mathit{l}(r_0) \frac{\sigma_{N}^2}{P_W}\biggr)  \exp\biggl(- \int_{\mathbb{R}^2} \frac{T\mathit{l}(r_0) \lambda_W h^{W}_{3}(r_0,x)}{\frac{P_L}{P_W}\mathit{l}(\|x\|) + T \mathit{l}(r_0)}  {\rm d}x \biggl)  \nonumber \\
	&\times \exp\biggl(  -\int_{\mathbb{R}^2 \setminus B(0,r_0)} \frac{T \mathit{l}(r_0) \lambda_L h_{3}^{L}(r_0,y)}{\mathit{l}(\|y\|) + T \mathit{l}(r_0)} {\rm d}y\biggr) f_{L}(r_0){\rm d}r_0.
	\end{align*}
\end{lemma}
\subsection{LTE with Lower Channel Access Priority as Wi-Fi when (a,b) = (1,2)} 
In this case, since the random backoff timer for each LTE eNB is always larger than that of Wi-Fi APs, LTE has a lower channel access priority. Specifically, the medium access indicator for each Wi-Fi AP and LTE eNB in~(\ref{MAIEqW4}) can be simplified to: 
\allowdisplaybreaks
\begin{align}\label{MAIEqW4Case2}
\allowdisplaybreaks
e_{i}^{W} = & \prod\limits_{x_{j} \in \Phi_{W} \setminus \{x_{i}\}} \biggl(\mathbbm{1} _{t_{j}^{W} \geq t_{i}^{W}} +\mathbbm{1} _{t_{j}^{W} < t_{i}^{W}} \mathbbm{1}_{ G_{ji}^{W}/ \mathit{l}(\|x_{j} - x_{i}\|) \leq \frac{\Gamma_{cs}}{P_W}} \biggl) , \nonumber\\
e_{k}^{L} = & \prod\limits_{x_{j} \in \Phi_{W}}  \mathbbm{1}_{ G_{jk}^{WL}/ \mathit{l}(\|x_{j}-y_{k}\|) \leq \frac{\Gamma^{L}}{P_W}}  \prod\limits_{y_{m} \in \Phi_{L} \setminus \{y_{k}\}} \biggl(\mathbbm{1}_{t_{m}^{L} \geq t_{k}^{L}} +\mathbbm{1} _{t_{m}^{L} < t_{k}^{L}} \mathbbm{1}_{ G_{mk}^{L}/ \mathit{l}(\|y_{m} - y_{k}\|) \leq \frac{\Gamma^{L}}{P_L}} \biggl).
\end{align}
The MAP for the typical AP and eNB are given in the following lemma.
\begin{lemma}
	When LTE follows a LBT and random BO mechanism with $(a,b) = (1,2)$, the MAPs for typical AP and eNB, denoted by $p_{4,MAP}^{W}$ and $p_{4,MAP}^{L}$ respectively, are given by:
	\allowdisplaybreaks
	\begin{align*}
	\allowdisplaybreaks
	p^{W}_{4,\text{MAP}}(\lambda_W,\lambda_L) &= \frac{1-\exp(-N^W)}{N^W},\\
	p^{L}_{4,\text{MAP}}(\lambda_W,\lambda_L) &= \exp(-N^W_3(\Gamma^{L}) )\frac{1-\exp(-N^L_3(\Gamma^{L}))}{N^L_3(\Gamma^{L})}.
	\end{align*}
\end{lemma}

\begin{corollary}
	When LTE follows a LBT and random BO mechanism with $(a,b) = (1,2)$, the MAP for the tagged Wi-Fi AP and LTE eNB are given by:
		\allowdisplaybreaks
		\begin{align*}
		\allowdisplaybreaks
		\hat{p}^{W}_{4,\text{MAP}}(\lambda_W,\lambda_L) &=  \int_{0}^{\infty} \frac{1-\exp(-N^W_2(r_0))}{N^W_2(r_0)}  f_{W}(r_0)  {\rm d}r_0,\\
		\hat{p}^{L}_{4,\text{MAP}}(\lambda_W,\lambda_L) &= \int_{0}^{\infty}\exp(-N^W_3(\Gamma^L)) \frac{1-\exp(-N^L_1(r_0,\Gamma^L))}{N^L_1(r_0,\Gamma^L)} f_{L}(r_0)  {\rm d}r_0.
		\end{align*}
\end{corollary}
In terms of MAP, this scheme is optimal to protect Wi-Fi since each AP has the same MAP as if LTE does not exist. In contrast, since eNBs will not transmit whenever a strong interfering AP exists, eNBs will have a role similar to APs in the scenario when LTE transmits continuously. 

In order to determine the coverage probability for the typical STA and typical UE, procedures similar to that of the previous parts are used. In particular, 
the conditional MAP $\mathbb{P}_{\Phi_W}^{x_i}(e_i^W = 1|e_0^W = 1, x_0 = (r_0,0))$, denoted by $h_4^{W}(r_0,x_i)$, can be directly obtained from Corollary~\ref{PWSelecProbW1Wi-FiCoro} by making $\lambda_L = 0$, which is given by:
\begin{align*}
\allowdisplaybreaks
h_4^{W}(r_0,x_i) = \frac{V(x-x_0,\frac{\Gamma_{cs}}{P_W},\frac{\Gamma_{cs}}{P_W},N^W_2(r_0),N^W_0(x,r_0,\Gamma_{cs}),C^W_2(x,x_0)) }
{U(x-x_0,\frac{\Gamma_{cs}}{P_W},N^W_2(r_0))}.
\end{align*}In addition, denote the conditional probability $\mathbb{P}_{\Phi_L}^{y_k} (e_{k}^{L} = 1 | e_{0}^{W} =1, x_0 = (r_0,0))$ by $h_4^{L}(r_0,y_k)$, we get:
\allowdisplaybreaks
\begin{align*}
\allowdisplaybreaks
h_4^{L}(r_0,y_k) = \frac{\frac{1-\exp(-N^L_{3}(\Gamma^{L}))}{N^L_3(\Gamma^{L})}\frac{1-\exp(-N^W_2(r_0) + C^W_0(y_k, \Gamma^{L},x_0,\Gamma_{cs}))}{N^W_2(r_0) - C^W_0(y_k, \Gamma^{L},x_0,\Gamma_{cs})} (1-\exp(-\mu \frac{\Gamma^{L}}{P_W}\mathit{l}(\|x_0-y_k\|)))}{\frac{1-\exp(-N^W_2(r_0))}{N^W_2(r_0)}\exp(N^W_0(y_k,r_0,\Gamma^{L}))}.
\end{align*}
By applying the non-homogeneous PPP approximation to Wi-Fi and LTE interferers, the
SINR coverage probability of the typical STA can be derived by the following procedures which are similar to that of
Lemma~\ref{CondiCOPW4Wi-FiLemma}:
\begin{lemma}\label{COPW4Case2Wi-FiLemma}
	When LTE implements the listen-before-talk and random backoff mechanism with $(a,b) = (1,2)$, the approximate SINR coverage probability of a typical Wi-Fi STA is:
	\allowdisplaybreaks
	\begin{align*}
	\allowdisplaybreaks
	p_4^W(T,\lambda_W,\lambda_L)  \approx & \int_{0}^{\infty}\exp\biggl( -\mu T \mathit{l}(r_0) \frac{\sigma_{N}^2}{P_W}\biggr) \exp\biggl(  -\int_{\mathbb{R}^2} \frac{T \mathit{l}(r_0) \lambda_L h_{4}^{L}(r_0,y)}{\frac{P_W}{P_L}\mathit{l}(\|y\|) + T \mathit{l}(r_0)} {\rm d}y\biggr) \nonumber \\
	&\times \exp\biggl(- \int_{\mathbb{R}^2 \setminus B(0,r_0)} \frac{T\mathit{l}(r_0) \lambda_W h^{W}_{4}(r_0,x)}{\mathit{l}(\|x\|) + T \mathit{l}(r_0)}   {\rm d}x \biggl) f_{W}(r_0) {\rm d}r_0.
	\end{align*}	
\end{lemma}
Next, given the tagged eNB of the typical UE is located at $y_0 = (r_0,0)$, the two conditional probabilities $\mathbb{P}_{\Phi_W}^{x_i}(e_{i}^W = 1 | e_{0}^{L} =1, y_0 = (r_0,0))$ and $\mathbb{P}_{\Phi_L}^{y_k}(e_{k}^L = 1 | e_{0}^{L} =1, y_0 = (r_0,0))$, denoted by $h_5^W(r_0,x_i)$ and $h_5^L(r_0,y_k)$ respectively, are given in~(\ref{CondiCOPCase2Wi-Fi2Eq}) and~(\ref{CondiCOPCase2LTE2Eq}):
\allowdisplaybreaks
\begin{align}
\allowdisplaybreaks
	&h_5^W(r_0,x_i) = \frac{1-\exp(-N^W + C^W_0(y_0-x_i,\Gamma^{L},o,\Gamma_{cs}))}{N^W - C^W_0(y_0-x_i,\Gamma^{L},o,\Gamma_{cs})} \label{CondiCOPCase2Wi-Fi2Eq},\\
     &h_5^L(r_0,y_k) =  \frac{(1-\exp(-\mu \frac{\Gamma^L}{P_L}\mathit{l}(\|y_k - y_0\|)) ) (M(N_4,N_5,N_6) + M(N_5,N_4,N_6))}{\exp(N^W_3(\Gamma^L)- C^W_0(y_0-y_k,\Gamma^L,o,\Gamma^L)) U(y_k-y_0,\frac{\Gamma^L}{P_L},N_4)}\label{CondiCOPCase2LTE2Eq},
\end{align}where $N_4 = N^L_1(r_0,\Gamma^{L})$, $N_5 = N^L_0(y_k,r_0,\Gamma^{L})$ and $N_6 = C^L_0(y_k,\Gamma^{L},y_0,\Gamma^{L})$ in~(\ref{CondiCOPCase2LTE2Eq}).

Based on $h_5^W$ and $h_5^L$, the SINR coverage probability of the typical UE can be derived by applying the non-homogeneous PPP approximation:
\begin{lemma}\label{COPLTEW4Case2Lemma}
	When LTE implements the listen-before-talk and random backoff mechanism with $(a,b) = (1,2)$, the approximate SINR coverage probability of the typical LTE UE is:
	\allowdisplaybreaks
	\begin{align*}
	\allowdisplaybreaks
	p_4^L(T,\lambda_W,\lambda_L)  \approx  &\int_{0}^{\infty} \exp\biggl( -\mu T \mathit{l}(r_0)\frac{\sigma_{N}^2}{P_W}\biggr) \exp\biggl(  -\int_{\mathbb{R}^2 \setminus B(0,r_0)} \frac{T \mathit{l}(r_0) \lambda_L h_{5}^{L}(r_0,y)}{\mathit{l}(\|y\|) + T \mathit{l}(r_0)} {\rm d}y\biggr) \\
	& \times \exp\biggl(- \int_{\mathbb{R}^2} \frac{T\mathit{l}(r_0) \lambda_W h^{W}_{5}(r_0,x)}{\frac{P_L}{P_W}\mathit{l}(\|x\|) + T \mathit{l}(r_0)}  {\rm d}x \biggl)f_{L} (r_0) {\rm d}r_0.
	\end{align*}		
\end{lemma}

The SINR coverage performance of the typical STA and UE under two LTE channel access priority schemes is plotted in Fig.~\ref{COPWiFiW4SimuvsTheoFig} and Fig.~\ref{COPLTEW4SimuvsTheoFig}, where the simulation results are obtained from the definition of SINR in~(\ref{SINRWiFiW1Eq}) and~(\ref{SINRLTEEq}). The accuracy of the approximations can be validated for various LTE sensing thresholds and AP/eNB densities. Since both Wi-Fi/LTE adopt the LBT and random backoff mechanism, a good overall SINR coverage probability can be achieved for Wi-Fi and LTE. In addition, given LTE contention window size $(a,b)$, both Wi-Fi STA and LTE UE can achieve better SINR performance with a more sensitive threshold $\Gamma^L$, which is due to less LTE interference. It can also be observed that when LTE has lower channel access priority, a less sensitive threshold $\Gamma^L$ is needed to obtain a similar Wi-Fi SINR performance as in the case when LTE has the same channel access priority as Wi-Fi.


 \begin{figure}[h]
 	\begin{subfigure}[b]{0.5\textwidth}
 		\includegraphics[height=2in, width=3.1in]{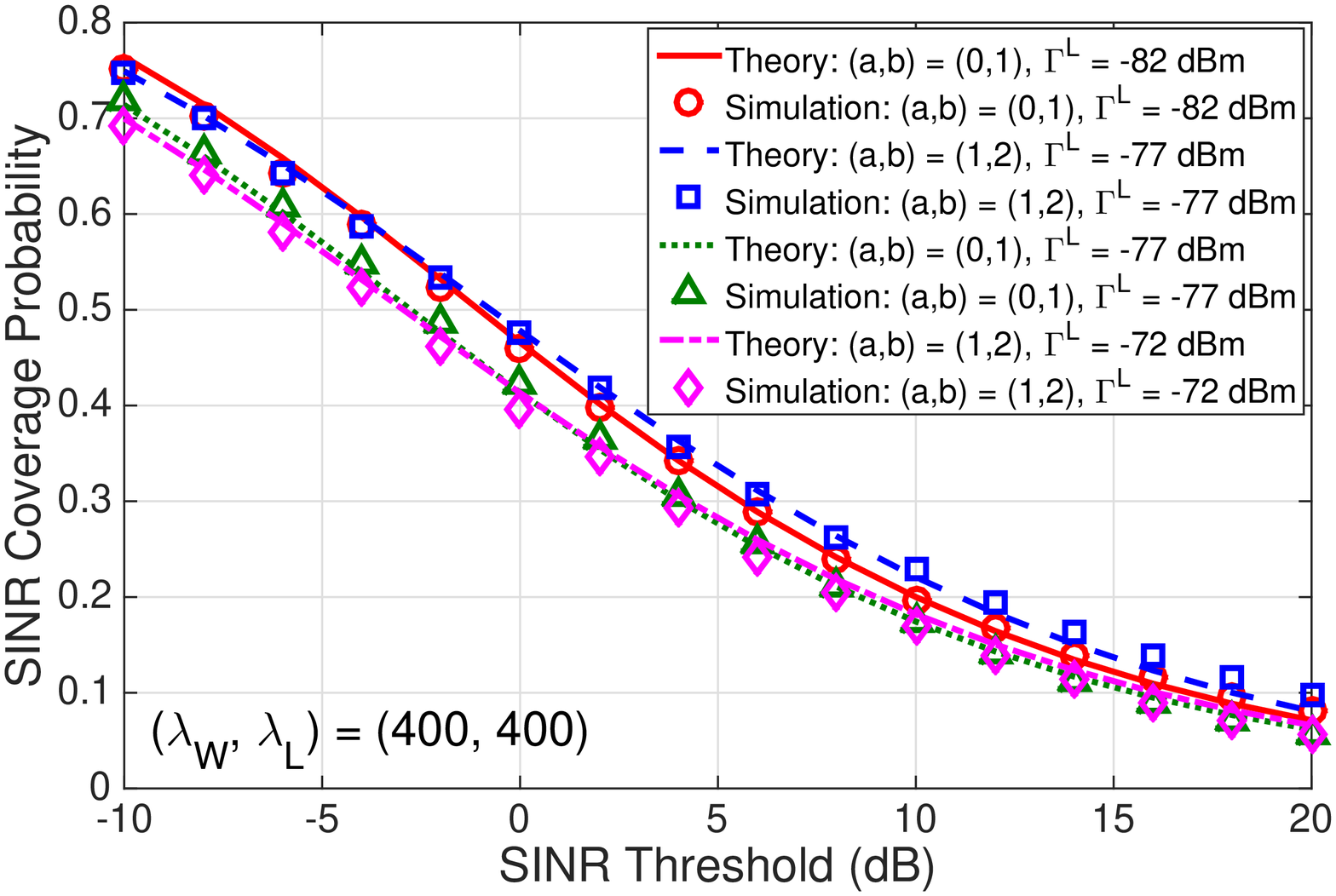}
 	\end{subfigure}
 	\hfill
 	\begin{subfigure}[b]{0.5\textwidth}
 		\includegraphics[height=2in, width=3.1in]{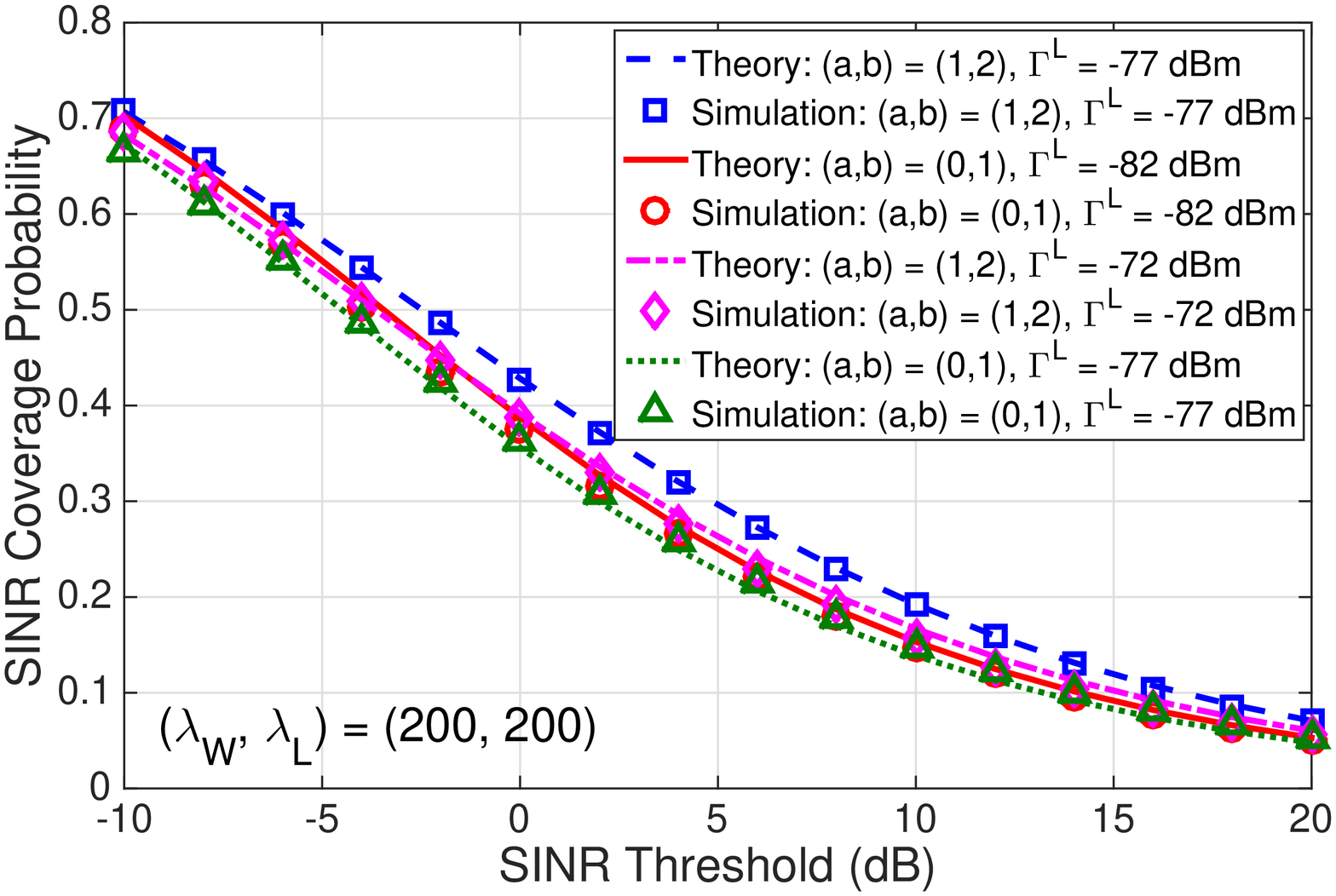}
 	\end{subfigure}
 	\caption{Wi-Fi SINR performance under different LTE channel access priorities and $\Gamma^L$.}\label{COPWiFiW4SimuvsTheoFig}
 \end{figure}

\begin{figure}[h]
	\begin{subfigure}[b]{0.5\textwidth}
		\includegraphics[height=2in, width=3.1in]{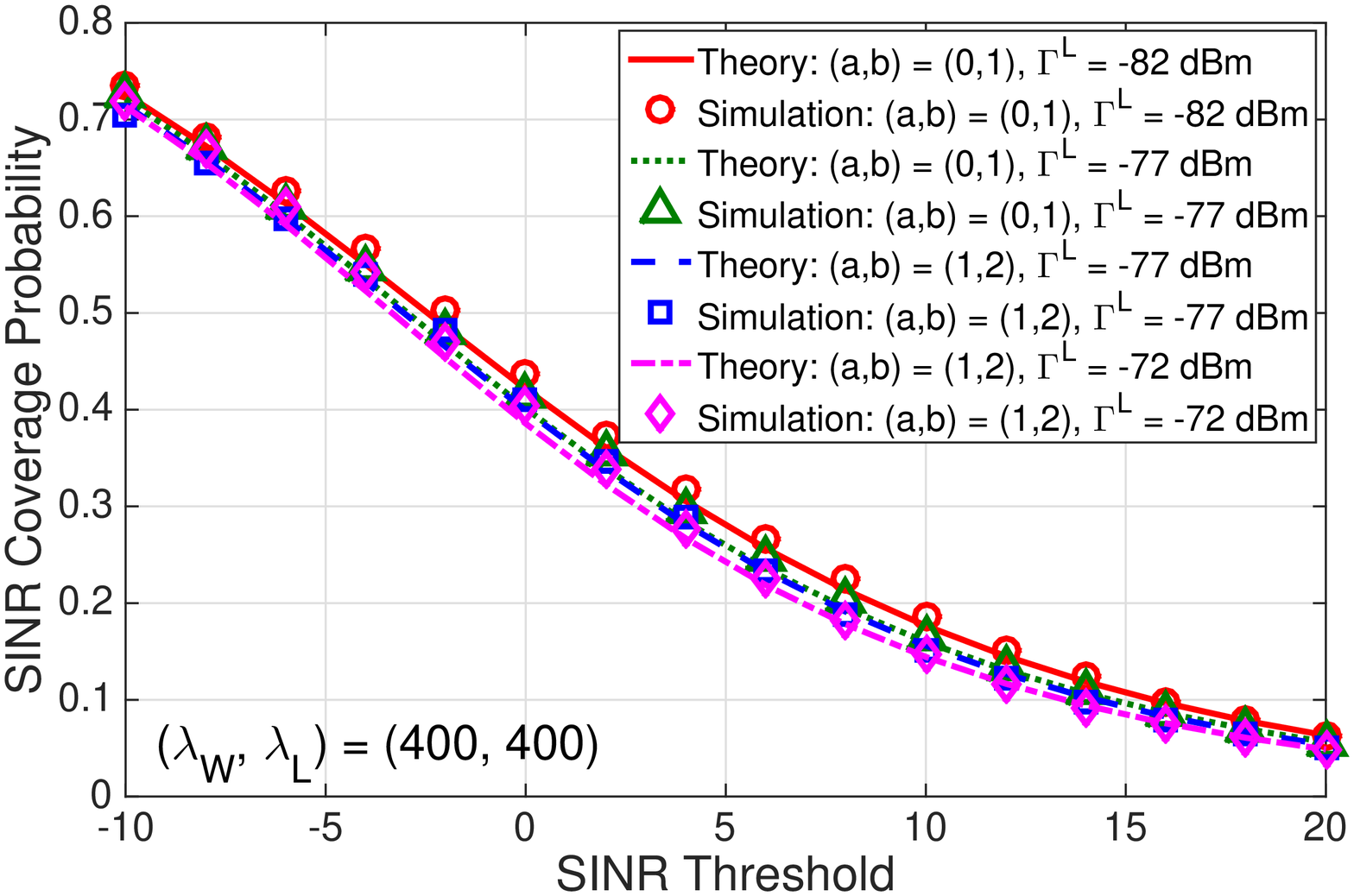}
	\end{subfigure}
	\hfill
	\begin{subfigure}[b]{0.5\textwidth}
		\includegraphics[height=2in, width=3.1in]{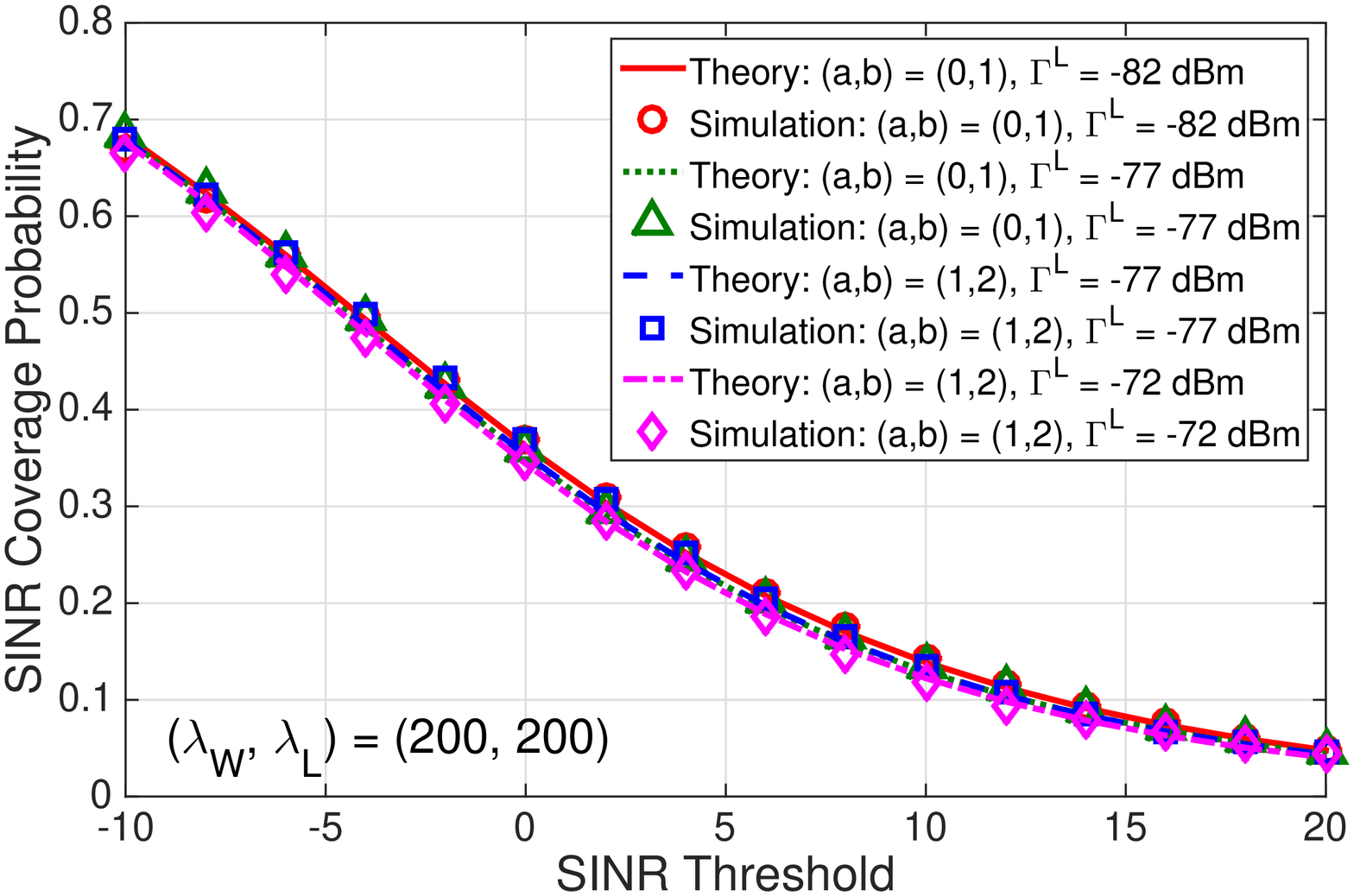}
	\end{subfigure}
	\caption{LTE SINR performance under different channel access priorities and $\Gamma^L$.}\label{COPLTEW4SimuvsTheoFig}
\end{figure}


\section{Performance Comparisons of Different Coexistence Scenarios}\label{NumEvaSec}
In this section, the DST and rate coverage performance for each coexistence scenario are compared through numerical evaluations. In particular, we use Wi-Fi + LTE (Wi-Fi + LTE-U, and Wi-Fi + LAA respectively) to denote the scenario when Wi-Fi operator 1 coexists with another operator 2, which uses LTE with no protocol change (LTE with discontinuous transmission, and LTE with LBT and random BO respectively). The baseline performance of Wi-Fi operator 1 is when operator 2 also uses Wi-Fi (i.e., Wi-Fi + Wi-Fi). The Wi-Fi MAP and SINR coverage of the baseline scenario can be obtained directly from Lemma~\ref{MAPW4Case1Lemma} and Lemma~\ref{CondiCOPW4Wi-FiLemma} by setting all the sensing thresholds to $\Gamma_{cs}$. In addition, we focus on a dense network deployment where $\lambda_W$ = 400 APs/km$^2$ and $\lambda_L$ = 400 eNBs/km$^2$. Based on the MAP and approximate SINR coverage probability, we have investigated the DST and rate coverage probability of Wi-Fi/LTE under all the coexistence scenarios in Fig.~\ref{DSTCompFig} and Fig.~\ref{RateCompFig}. 

\begin{figure}
	\begin{subfigure}[b]{0.5\textwidth}
		\includegraphics[height=2in, width=3.1in]{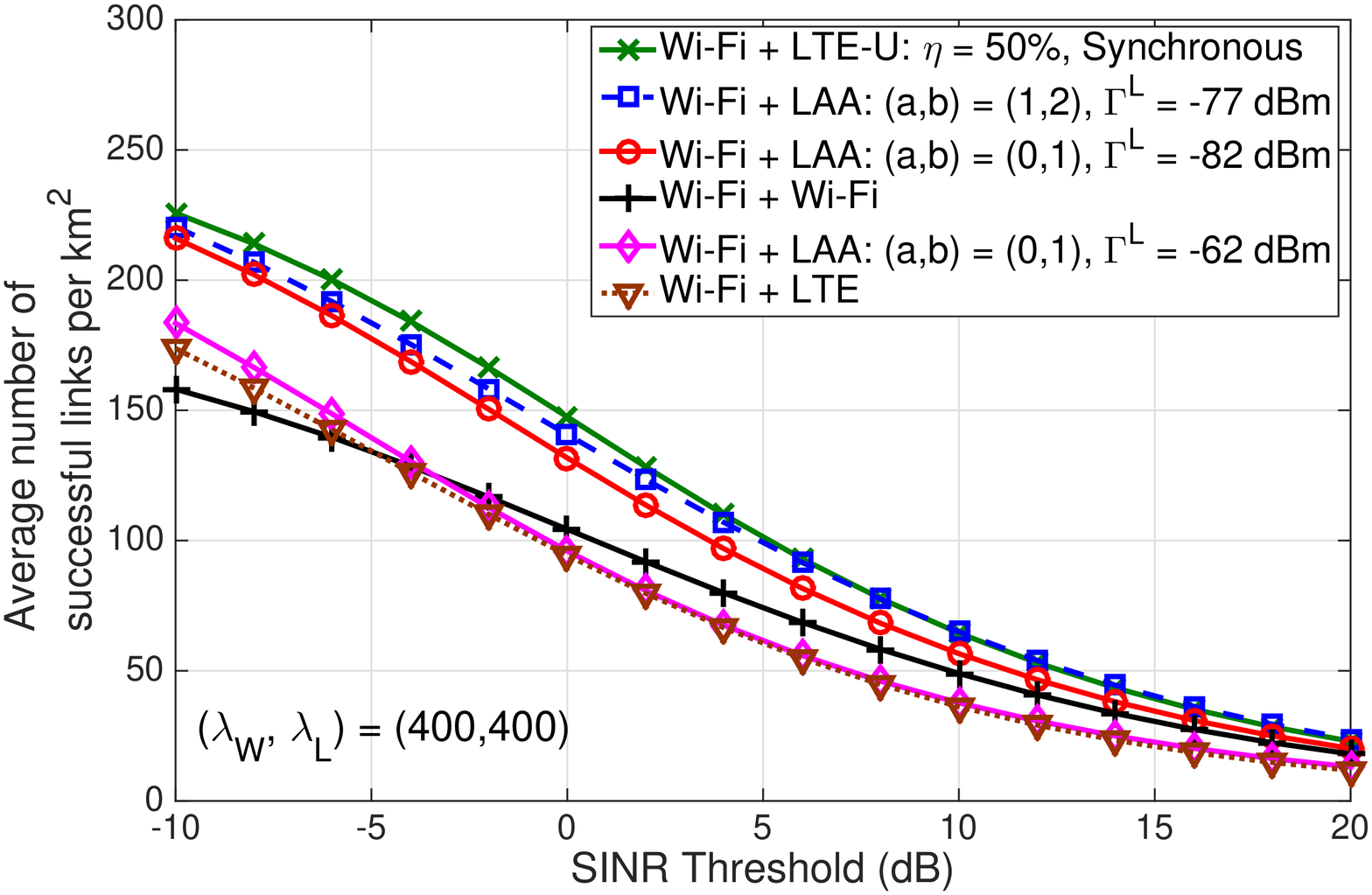}
		\caption{DST for Wi-Fi}\label{DSTCompFigWifi}
	\end{subfigure}
	\hfill
	\begin{subfigure}[b]{0.5\textwidth}
		\includegraphics[height=2in, width=3.1in]{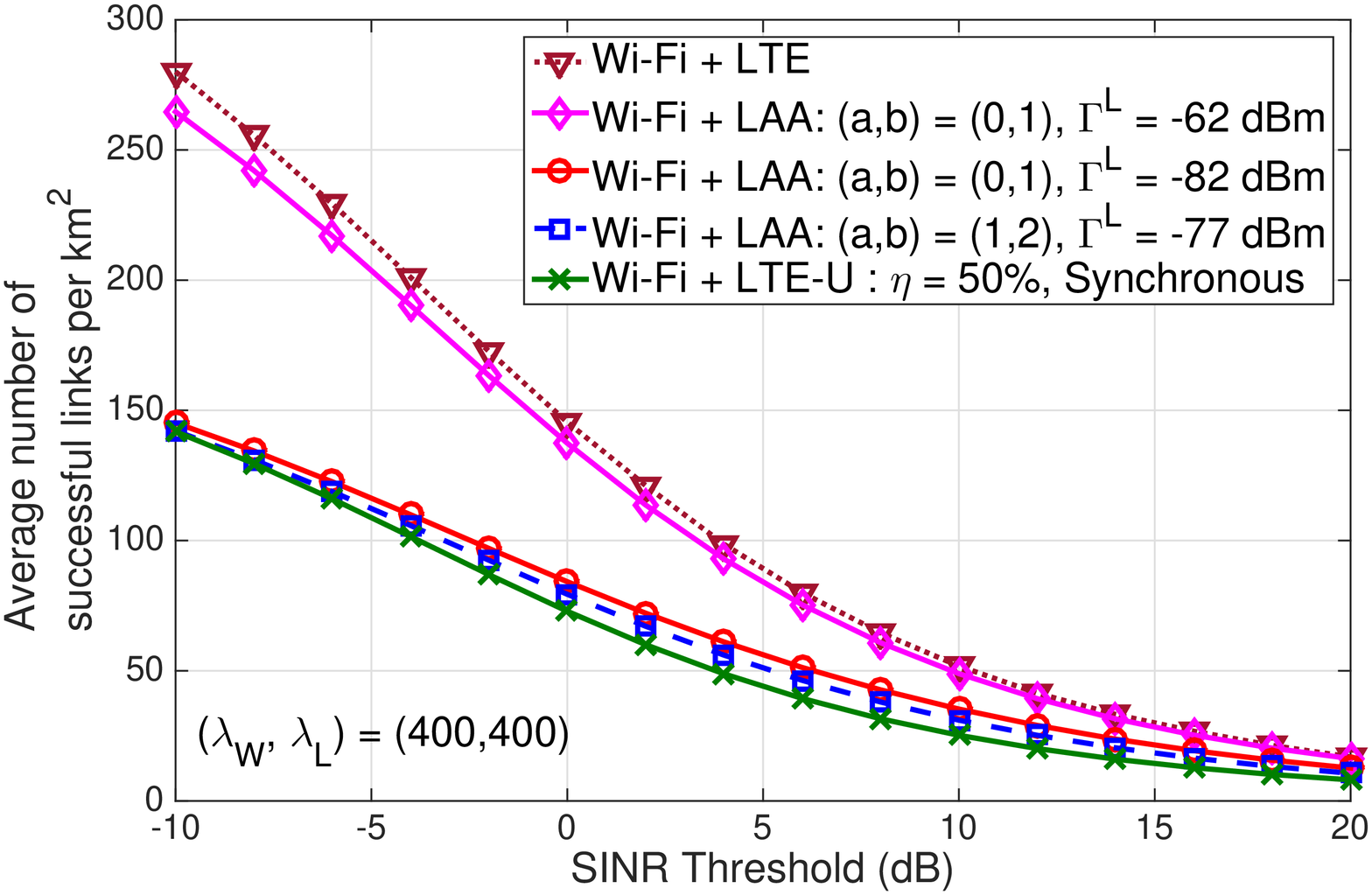}
		\caption{DST for LTE}\label{DSTCompFigLTE}
	\end{subfigure}
	\caption{DST comparisons under different coexistence scenarios.}\label{DSTCompFig}
\end{figure}

\begin{figure}
	\begin{subfigure}[b]{0.5\textwidth}
		\includegraphics[height=2in, width=3.1in]{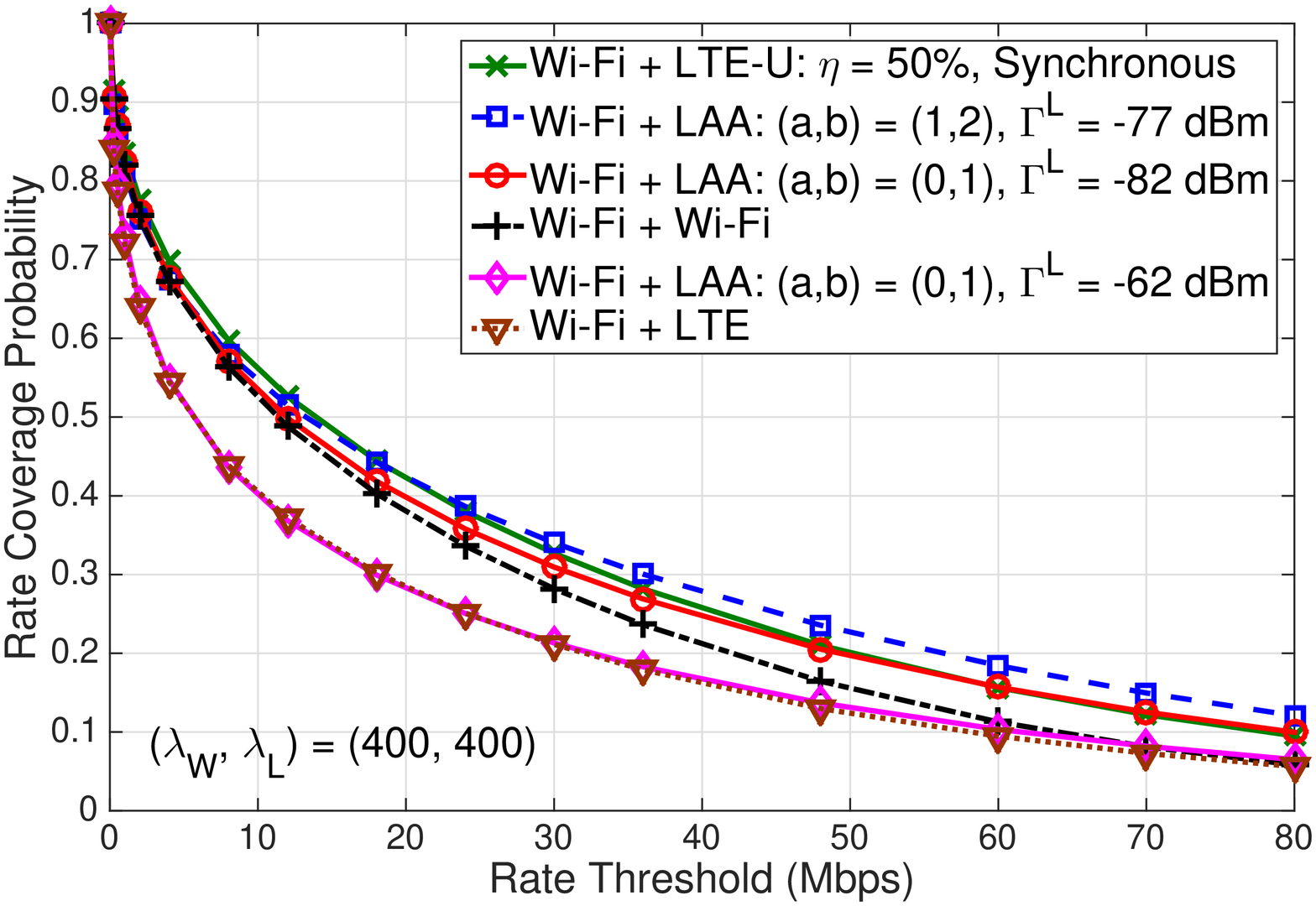}
		\caption{Rate coverage for Wi-Fi}\label{RateCompFigWifi}
	\end{subfigure}
	\hfill
	\begin{subfigure}[b]{0.5\textwidth}
		\includegraphics[height=2in, width=3.1in]{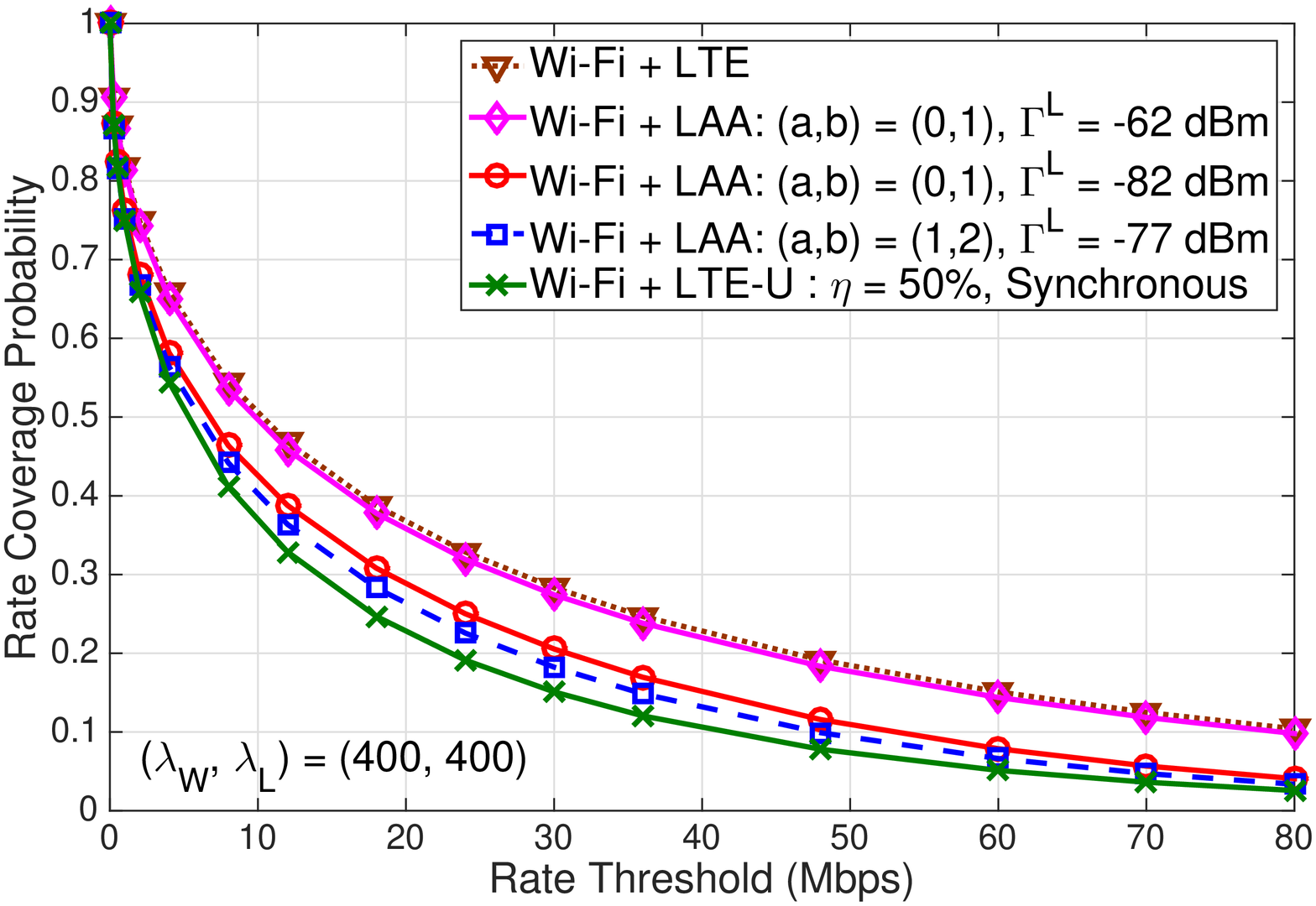}
		\caption{Rate coverage for LTE}\label{RateCompFigLTE}
	\end{subfigure}
	\caption{Rate coverage comparisons under different coexistence scenarios.}\label{RateCompFig}
\end{figure}

Fig.~\ref{DSTCompFigWifi} shows that when coexisting with LTE, Wi-Fi has the worst DST performance since it experiences strong interference from the persistent transmitting LTE eNBs. In addition, Wi-Fi achieves similar DST performance when operator 2 implements one of the following mechanisms: (1) LTE-U with a short duty cycle (e.g., 50\%); (2) LAA with same channel access priority as Wi-Fi and a more sensitive sensing threshold (e.g., $(a,b) = (0,1)$, $\Gamma^L = -82$ dBm); and (3) LAA with lower channel access priority than Wi-Fi and a less sensitive sensing threshold (e.g., $(a,b) = (1,2)$, $\Gamma^L = -77$ dBm). 
Compared to the baseline scenario, Wi-Fi has better DST under the above scenarios, especially in the low SINR threshold regime. 
Furthermore, when operator 2 implements LAA with the -62 dBm energy detection threshold, the DST performance of Wi-Fi is not much improved over Wi-Fi + LTE. Therefore, the -62 dBm energy detection threshold is too conservative to protect Wi-Fi, and a more sensitive threshold $\Gamma^L$ is recommended for LAA. 
In contrast, Fig.~\ref{DSTCompFigLTE} shows that operator 2 has significantly lower (around 50\%) DST when using LTE-U or LAA with a sensitive sensing threshold (e.g., -82 dBm or -77 dBm), which is mainly due to the decreased MAP for eNBs.

In terms of rate coverage, it can be observed from Fig.~\ref{RateCompFigWifi} that when operator 2 adopts LTE-U with a 50\% duty cycle or LAA with a sensitive sensing threshold (e.g., -82 dBm or -77 dBm), Wi-Fi has similar performance as the baseline scenario in low rate threshold regime (e.g. less than 5 Mbps), and better performance with medium to high rate threshold (e.g., more than 10 Mbps). If LAA uses the -62 dBm energy detection threshold, the rate coverage of Wi-Fi has negligible improvement over the Wi-Fi + LTE scenario, which means the energy detection threshold does not suffice to protect Wi-Fi. In addition, due to the degraded SINR performance, Wi-Fi has worse rate performance under Wi-Fi + LTE than the baseline scenario. Meanwhile, when the sensing threshold of LAA is -82 dBm or -77 dBm, Fig.~\ref{RateCompFigLTE} shows that the rate loss of operator 2 under Wi-Fi + LAA is around 30\% to 40\% compared to Wi-Fi + LTE. In contrast, the rate loss under Wi-Fi + LTE-U is slightly more than 50\% for most rate thresholds.

Overall, under Wi-Fi + LTE, the DST and rate coverage probability of Wi-Fi decreases significantly compared to the baseline performance, which makes it an impractical scenario to operate LTE in unlicensed spectrum. Under Wi-Fi + LTE-U, LTE-U operator 2 has the flexibility to guarantee good DST and rate coverage performance for Wi-Fi operator 1 by choosing a low LTE transmission duty cycle. In addition, LTE-U has low implementation cost due to its simple scheme. However, LTE-U also has the following disadvantages: (1) the LTE-U operator has much degraded DST and rate coverage performance under low transmission duty cycle; (2) LTE-U is only feasible in certain regions and/or unlicensed bands where LBT feature is not required, such as the 5.725-5.825 GHz band in U.S.~\cite{3gpp2015TR}; (3) LTE-U transmissions are more likely to collide with Wi-Fi acknowledgment packets due to the lack of a CCA procedure, which means Wi-Fi SINR coverage under Wi-Fi + LTE-U may not as easily translate into the rate performance as Wi-Fi + LAA; and (4) LTE-U has practical cross-layer issues, such that the frequent on and off switching of LTE will trigger the Wi-Fi rate control algorithm to lower the Wi-Fi transmission rate~\cite{jindal2015LTEWiFi}. 
In contrast, under Wi-Fi + LAA, by choosing an appropriate LAA channel access priority (i.e., contention window size) and sensing threshold, Wi-Fi operator 1 also achieves better DST and rate coverage performance compared to the baseline scenario, while LAA operator 2 can maintain acceptable rate coverage performance. Additionally, LAA also meets the global requirement for operation in the unlicensed spectrum. The main disadvantage of LAA versus LTE-U is that LAA requires more complicated implementation for the LBT and random BO feature. 
Therefore, in terms of performance comparisons and practical constraints, LTE with LBT and random BO (i.e., LAA) is more promising than LTE with discontinuous transmission (i.e., LTE-U) to provide a global efficient solution to the coexistence issues of LTE and Wi-Fi in unlicensed spectrum. 


\section{Conclusion}\label{SecSum}
This paper proposed and validated a stochastic geometry framework for analyzing the coexistence of overlaid Wi-Fi and LTE networks. Performance metrics including the medium access probability, SINR coverage probability, density of successful transmission and rate coverage probability have been analytically derived and numerically evaluated under three coexistence scenarios. Although Wi-Fi performance is significantly degraded when LTE transmits continuously without any protocol changes, we showed that LTE can be a good neighbor to Wi-Fi by manipulating the LTE transmission duty cycle, sensing threshold, or channel access priority. The proposed analytical framework validates and complements the ongoing system level simulation studies of LTE-U and LAA, and can be utilized by both academia and industry to rigorously study LTE and Wi-Fi coexistence issues.
Future works could include: (1) extending the single unlicensed band analysis into multiple bands, and incorporating channel selection schemes to further improve Wi-Fi/LTE performance; (2) charactering Wi-Fi/LTE delay performance by extending the full-buffer traffic assumption to non full-buffer traffic; (3) analyzing the SINR or rate coverage when both downlink and uplink traffic exist; and (4) investigating Wi-Fi/LTE performance when LTE adopts both listen-before-talk and discontinuous transmission features, which will be standardized by LTE-U forum~\cite{lteu2015TRv1} in the future.
\appendices
\section{Proof of Corollary~\ref{PWSelecProbW1Wi-FiCoro}}\label{CondiCOPWiFiAppdx}
For AP $x_i \in \Phi_W$, the conditional MAP of $x_i$ given the tagged AP $x_0 = (r_0,0)$ transmits is:
\allowdisplaybreaks
\begin{align}
\allowdisplaybreaks
&\mathbb{P}_{\Phi_W}^{x_i}(e_i^W = 1 | e_0^W = 1, x_0 \in \Phi_W , \Phi_W(B^o(0,r_0)) = 0) \nonumber\\
\overset{(a)}{=}&\frac{{P}_{\Phi_W}^{x_i,x_0}(e_i^W = 1 , e_0^W = 1 | \Phi_W(B^o(0,r_0)) = 0)}{{P}_{\Phi_W}^{x_i,x_0}(e_0^W = 1 | \Phi_W(B^o(0,\|x_0\|)) = 0)} \nonumber\\
\overset{(b)}{=}&\frac{{E}_{\Phi_W}^{x_i}(\hat{e}_i^W \hat{e}_0^W )}{{E}_{\Phi_W}^{x_i}(\hat{e}_0^W )},\label{PWSelecProbW1Wi-FiPfEq1}
\end{align}where (a) follows from the Baye's rule, and (b) is derived by applying the Slyvniak's theorem and de-conditioning. The modified medium access indicators for $x_i$ and $x_0$ are:
\allowdisplaybreaks
\begin{align}\label{MAIW1Wi-FiEq}
\allowdisplaybreaks
\hat{e}_i^W = &\prod\limits_{x_{j} \in (\Phi_{W}\cap B^c(0,r_0) +\delta_{x_0}) \setminus \{x_{i}\}} \!\!\!\biggl(\mathbbm{1} _{t_{j}^{W} \geq t_{i}^{W}} +\mathbbm{1} _{t_{j}^{W} < t_{i}^{W}} \mathbbm{1}_{ G_{ji}^{W}/ \mathit{l}(\|x_{j}-x_{i}\|) \leq \frac{\Gamma_{cs}}{P_W}} \biggl) \prod\limits_{y_{k} \in \Phi_{L}} \mathbbm{1}_{G_{ki}^{LW}/ \mathit{l}(\|y_{k} - x_{i}\|) \leq \frac{\Gamma_{ed}}{P_L }}, \nonumber\\
\hat{e}_0^W = &\prod\limits_{x_{j} \in \Phi_{W} \cap B^c(0,r_0) } \biggl(\mathbbm{1}_{t_{j}^{W} \geq t_{0}^{W}} +\mathbbm{1} _{t_{j}^{W} < t_{0}^{W}} \mathbbm{1}_{ G_{j0}^{W}/ \mathit{l}(\|x_{j}-x_{0}\|) \leq \frac{\Gamma_{cs}}{P_W}} \biggl) \prod\limits_{y_{k} \in \Phi_{L}} \mathbbm{1}_{G_{k0}^{LW}/ \mathit{l}(\|y_{k} - x_{0}\|) \leq \frac{\Gamma_{ed}}{P_L }}.
\end{align}
Therefore, the denominator in~(\ref{PWSelecProbW1Wi-FiPfEq1}) is given by:
\allowdisplaybreaks
\begin{align}
\allowdisplaybreaks
&\mathbb{E}_{\Phi_W}^{x_i}\left[\prod\limits_{x_{j} \in \Phi_{W} \cap B^c(0,r_0) } \biggl(\mathbbm{1}_{t_{j}^{W} \geq t_{0}^{W}} +\mathbbm{1} _{t_{j}^{W} < t_{0}^{W}} \mathbbm{1}_{ G_{j0}^{W}/ \mathit{l}(\|x_{j}-x_{0}\|) \leq \frac{\Gamma_{cs}}{P_W}} \biggl) \prod\limits_{y_{k} \in \Phi_{L}} \mathbbm{1}_{G_{k0}^{LW}/ \mathit{l}(\|y_{k} - x_{0}\|) \leq \frac{\Gamma_{ed}}{P_L }}\right] \nonumber\\
=&\int_{0}^{1} \mathbb{E}\left[\prod\limits_{x_{j} \in (\Phi_{W} \cap B^c(0,r_0)+ \delta_{x_i})  } \!\!\!\!\!\!\biggl(\mathbbm{1}_{t_{j}^{W} \geq t} +\mathbbm{1} _{t_{j}^{W} < t} \mathbbm{1}_{ G_{j0}^{W}/ \mathit{l}(\|x_{j}-x_{0}\|) \leq \frac{\Gamma_{cs}}{P_W}} \biggl) \biggl| t_0^W = t\right]{\rm d}t \mathbb{E}\left[\prod\limits_{y_{k} \in \Phi_{L}} \mathbbm{1}_{G_{k0}^{LW}/ \mathit{l}(\|y_{k} - x_{0}\|) \leq \frac{\Gamma_{ed}}{P_L }}\right] \nonumber\\
=&U(x_i-x_0,\frac{\Gamma_{cs}}{P_W},N^W_2(r_0)) \exp(-N^L). \label{PWSelecProbW1Wi-FiPfDen}
\end{align}
On the other hand, the numerator in~(\ref{PWSelecProbW1Wi-FiPfEq1}) can be computed as:
\allowdisplaybreaks
\begin{align}
\allowdisplaybreaks
&\mathbb{E}_{\Phi_W}^{x_i}(\hat{e}_i^W \hat{e}_0^W ) \nonumber \\
\overset{(a)}{=}&\mathbb{E}\biggl[\prod\limits_{x_{j} \in (\Phi_{W}\cap B^c(0,r_0)+\delta_{x_0})  } \biggl(\mathbbm{1}_{t_{j}^{W} \geq t_{i}^{W}} +\mathbbm{1} _{t_{j}^{W} < t_{i}^{W}} \mathbbm{1}_{ G_{ji}^{W}/ \mathit{l}(\|x_{j}-x_{i}\|) \leq \frac{\Gamma_{cs}}{P_W}} \biggl) \prod\limits_{y_{k} \in \Phi_{L}} \mathbbm{1}_{G_{ki}^{LW}/ \mathit{l}(\|y_{k} - x_{i}\|) \leq \frac{\Gamma_{ed}}{P_L }}  \nonumber\\
&\times \prod\limits_{x_{j} \in (\Phi_{W}\cap B^c(0,r_0) +\delta_{x_i}) } \biggl(\mathbbm{1}_{t_{j}^{W} \geq t_{0}^{W}} +\mathbbm{1} _{t_{j}^{W} < t_{0}^{W}} \mathbbm{1}_{ G_{j0}^{W}/ \mathit{l}(\|x_{j}-x_{0}\|) \leq \frac{\Gamma_{cs}}{P_W}} \biggl) \prod\limits_{y_{k} \in \Phi_{L}} \mathbbm{1}_{G_{k0}^{LW}/ \mathit{l}(\|y_{k} - x_{0}\|) \leq \frac{\Gamma_{ed}}{P_L }}\biggl]\nonumber \\
=& \int_{0}^{1} \int_{0}^{1} \mathbb{E}\biggl[\prod\limits_{x_j \in \Phi_W \cap B^c(0,r_0)} (1-\mathbbm{1}_{t_j < t} \mathbbm{1}_{ G_{j0}^{W}/ \mathit{l}(\|x_{j}-x_{0}\|) > \frac{\Gamma_{cs}}{P_W}}) (1-\mathbbm{1}_{t_j < t^{'}} \mathbbm{1}_{ G_{ji}^{W}/ \mathit{l}(\|x_{j}-x_{i}\|) > \frac{\Gamma_{cs}}{P_W}}) \nonumber \\
& \times \mathbbm{1}_{ G_{0i}^{W}/ \mathit{l}(\|x_{0}-x_{i}\|) \leq \frac{\Gamma_{cs}}{P_W}} \prod\limits_{y_{k} \in \Phi_{L}}\mathbbm{1}_{G_{k0}^{LW}/ \mathit{l}(\|y_{k} - x_{0}\|) \leq \frac{\Gamma_{ed}}{P_L }} \mathbbm{1}_{G_{ki}^{LW}/ \mathit{l}(\|y_{k} - x_{i}\|) \leq \frac{\Gamma_{ed}}{P_L }} \biggl| t_0^W = t, t_i^W = t^{'} \biggr]{\rm d}t^{'}{\rm d}t \nonumber\\
=&\exp(-2 N^L+ C^L_2(x_i-x_0)) V(x_i-x_0,\frac{\Gamma_{cs}}{P_W},\frac{\Gamma_{cs}}{P_W},N^W_2(r_0),N^W_0(x_i,r_0,\Gamma_{cs}),C^W_1(x_i,x_0)),\nonumber 
\end{align}
where (a) follows from Slyvniak's theorem. 
\section{Proof of Corollary~\ref{CondiRetainProbW4Wi-FiLemma}}\label{CondiMAPW4Appdx}
For every AP $x_i$, the quantity that needs to be computed is $h_2^W(r_0,x_i) = \mathbb{P}_{\Phi_W}^{x_i} [e_i^W = 1 |e_0^{W} = 1,x_0 = (r_0,0)]$. Similar to~(\ref{PWSelecProbW1Wi-FiPfEq1}), $h_2^W(r_0,x_i)$ can be rewritten as $\frac{\mathbb{E}_{\Phi_W}^{x_i}(\hat{e}_i^W\hat{e}_0^W)}{\mathbb{E}_{\Phi_W}^{x_i}(\hat{e}_0^W)}$, where:
\allowdisplaybreaks
\begin{align*}
\allowdisplaybreaks
\hat{e}_{i}^{W} = & \!\!\!\!\!\!\!\!\!\!\!\!\prod\limits_{x_{j} \in (\Phi_{W} \cap B^c(0,r_0) + \delta_{x_0}) \setminus \{x_{i}\}} \!\!\!\!\!\!\!\!\!\!\!\!\biggl(\mathbbm{1} _{t_{j}^{W} \geq t_{i}^{W}} +\mathbbm{1} _{t_{j}^{W} < t_{i}^{W}} \mathbbm{1}_{G_{ji}^{W}/ \mathit{l}(\|x_{j} - x_{i}\|) \leq \frac{\Gamma_{cs}}{P_W }} \biggl) \prod\limits_{y_{m} \in \Phi_{L}} \biggl(\mathbbm{1} _{t_{m}^{L} \geq t_{i}^{W}} +\mathbbm{1} _{t_{m}^{L} < t_{i}^{W}} \mathbbm{1}_{G_{mi}^{LW}/ \mathit{l}(\|y_{m} - x_{i}\|) \leq \frac{\Gamma_{ed}}{P_L }} \biggl),\nonumber\\
\hat{e}_{0}^{W} = & \!\!\!\!\!\!\!\!\prod\limits_{x_{j} \in \Phi_{W} \cap B^c(0,r_0) }\!\!\biggl(\mathbbm{1} _{t_{j}^{W} \geq t_{0}^{W}} +\mathbbm{1} _{t_{j}^{W} < t_{0}^{W}} \mathbbm{1}_{G_{j0}^{W}/ \mathit{l}(\|x_{j} - x_{0}\|) \leq \frac{\Gamma_{cs}}{P_W}} \biggl)  \prod\limits_{y_{m} \in \Phi_{L}} \biggl(\mathbbm{1} _{t_{m}^{L} \geq t_{0}^{W}} +\mathbbm{1} _{t_{m}^{L} < t_{0}^{W}} \mathbbm{1}_{G_{m0}^{LW}/ \mathit{l}(\|y_{m} - x_{0}\|) \leq \frac{\Gamma_{ed}}{P_L}} \biggl).
\end{align*}
Both $\mathbb{E}_{\Phi_W}^{x_i}(\hat{e}_0^W)$ and $\mathbb{E}_{\Phi_W}^{x_i}(\hat{e}_i^W\hat{e}_0^W)$ can be calculated using Slyvniak's theorem and the PGFL of PPP, which will give the result in~(\ref{CondiRetainProbW4Wi-FiEq}).
\bibliographystyle{ieeetr}
\bibliography{reference}

\end{document}